\documentclass[12pt]{article}
\usepackage[margin=1in]{geometry}
\usepackage{amsmath}
\usepackage{graphicx}
\usepackage{hyperref}
\usepackage{enumerate}
\usepackage{natbib}
\usepackage{url} 
\usepackage{mathtools}
\usepackage{empheq}
\usepackage{amsfonts}
\usepackage{amsgen,amsmath,amstext,amsbsy,amsopn,amssymb,amsthm,stmaryrd,bbm,subfigure}
\usepackage{cprotect}
\usepackage{natbib}
\usepackage[dvipsnames]{xcolor}
\usepackage{comment}
\usepackage{enumitem}
\usepackage{multirow,booktabs,makecell}
\usepackage[affil-it]{authblk}
\usepackage{algorithmic}
\usepackage[ruled,vlined,longend]{algorithm2e}

\newcommand{\RNum}[1]{\uppercase\expandafter{\romannumeral #1\relax}}

\newcommand{\vertiii}[1]{{\left\vert\kern-0.25ex\left\vert\kern-0.25ex\left\vert #1 
		\right\vert\kern-0.25ex\right\vert\kern-0.25ex\right\vert}}

\newcommand{\bpi}{\boldsymbol{\pi}}
\newcommand{\bZ}{\mathbf{Z}}
\newcommand{\bI}{\mathbf{I}}
\newcommand{\bS}{\mathbf{S}}
\newcommand{\bW}{\mathbf{W}}
\newcommand{\bX}{\mathbf{X}}
\newcommand{\bY}{\mathbf{Y}}
\newcommand{\by}{\mathbf{y}}
\newcommand{\bx}{\mathbf{x}}
\newcommand{\bm}{\mathbf{m}}
\newcommand{\bth}{\boldsymbol{\theta}}
\newcommand{\btau}{\boldsymbol{\tau}}
\newcommand{\bet}{\boldsymbol{\beta}}
\newcommand{\bal}{\boldsymbol{\alpha}}
\newcommand{\bLam}{\mathbf{\Lambda}}
\newcommand{\bmu}{\boldsymbol{\mu}}
\newtheorem{thm}{Theorem}
\newtheorem{lem}{Lemma}
\newtheorem{rmk}{Remark}
\newtheorem{cor}{Corollary}

\newtheorem{assump}{Assumption}

\begin{document}
\title{\vspace{-1cm} Bend to Mend: Toward Trustworthy Variational Bayes with Valid Uncertainty Quantification}

\author[1]{Jiaming Liu}
\author[1]{Meng Li\thanks{Corresponding author: meng@rice.edu}}

\affil[1]{Department of Statistics, Rice University}

\date{}
\maketitle
\begin{abstract}
Variational Bayes (VB) is a popular and computationally efficient method to approximate the posterior distribution in Bayesian inference, especially when the exact posterior is analytically intractable and sampling-based approaches are computationally prohibitive. While VB often yields accurate point estimates, its uncertainty quantification (UQ) is known to be unreliable. For example, credible intervals derived from VB posteriors tend to exhibit undercoverage, failing to achieve nominal frequentist coverage probabilities. In this article, we address this challenge by proposing Trustworthy Variational Bayes (TVB), a method to recalibrate the UQ of broad classes of VB procedures. Our approach follows a bend-to-mend strategy: we intentionally misspecify the likelihood (\textit{bend}) to correct VB's flawed UQ (\textit{mend}).  In particular, we first relax VB by 
building on a recent fractional VB method indexed by a fraction $\omega$, and then identify the optimal fraction parameter using conformal techniques such as sample splitting and bootstrapping. This yields recalibrated UQ for any given parameter of interest. On the theoretical front, we establish that the calibrated credible intervals achieve asymptotically correct frequentist coverage for a given parameter of interest; this, to the best of our knowledge, is the first such theoretical guarantee for VB. On the practical front, we introduce the ``TVB table'', which enables (1) massive parallelization and remains agnostic to the parameter of interest during its construction, and (2) efficient identification of the optimal fraction parameter for any specified parameter of interest. The proposed method is illustrated via Gaussian mixture models and Bayesian mixture linear regression models, and numerical experiments demonstrate that the TVB method outperforms standard VB and achieves normal frequentist coverage in finite samples. A real data application is also discussed. 
\end{abstract}

\noindent
{\it Keywords:}  Variational Bayes, fractional posterior, sample splitting bootstrap, credible interval calibration

\newpage

\section{Introduction}
\label{intro}
 
Variational Bayes (VB) has emerged as a powerful method for Bayesian inference particularly when the exact posterior distribution is not available and Markov Chain Monte Carlo (MCMC) becomes computationally too intensive. VB reformulates Bayesian inference as an optimization problem \citep{wainwright2008graphical, blei2017variational, hoffman2013stochastic, tran2021practical} and can achieve results orders of magnitude faster than traditional MCMC methods, making it particularly attractive for applications in high-dimensional settings and complex modeling. In addition, we note that numerical convergence in VB is often easier to diagnose than the distributional convergence encountered in MCMC.

Uncertainty quantification (UQ) enables researchers to assess the reliability of their estimates and make informed decisions, and is often viewed as a key advantage of Bayesian inference. However, it is well known that the UQ yielded by VB might be problematic and lack trustworthiness. This limitation stems from the inherent bias introduced by the variational approximation and the potential underestimation of posterior variance \citep{wang2005inadequacy, bishop2006pattern, wang2019frequentist}.
While VB's limitation in UQ is well documented and acknowledged, improving its UQ is challenging. 
Existing work includes improved variance and covariance estimates \citep{giordano2018covariances}; \cite{syring2019calibrating} proposed a recalibration approach for general Gibbs posteriors, but their focus is not on VB and theoretical guarantees were not established. We are not aware of any existing work that achieves recalibrated interval estimates for VB with theoretical guarantees. The current article contributes to filling this gap. 

In this paper, we aim to recalibrate VB such that the credible intervals attain the nominal coverage level, leading to new approximate Bayesian methods that build on the computational advantage of the traditional VB while maintaining guaranteed frequentist properties on UQ. We propose a strategy called ``bend to mend'' to achieve trustworthy VB, where we first consider a broader class of approximations and then choose one from this broader class for recalibration. In particular, we use recently developed fractional VB which adjusts the scale of the posterior distribution via a fractional parameter $\omega$~\citep{yang2020alpha}. Such fractional VB approximates the fractional posterior distribution, in which we intentionally misspecify the likelihood function by raising it to a power. Then we select this $\omega$ parameter such that frequentist coverage is attained. Our approach differs from the General Posterior Calibration method of \cite{syring2019calibrating} both algorithmically and theoretically, with a specific focus on variational Bayes. In particular, the fractional VB we employ extends existing tempered posteriors or the Gibbs posterior considered in \cite{syring2019calibrating} when the models involve latent variables. More importantly, we propose to use a simple but powerful sample splitting approach \citep{10.1214/18-AOS1784, wasserman2020universal} and the flexibility of bootstrap to develop a computationally efficient approach to choose $\omega$. These differences turn out to be crucial for theoretical developments and finite-sample performance.

 The employed fractional VB can typically be implemented by minor modifications of existing algorithms to the traditional VB, facilitating algorithmic development. Our theory shows that fractional VB with our selected $\omega$ achieves guaranteed nominal frequentist coverage, leading to a class of approximate Bayesian methods building on VB yet with justified UQ that is not available for the traditional VB. To the best of our knowledge, this is the first such result that provides a theoretical guarantee when recalibrating the UQ produced by variational Bayesian methods. 

The rest of this paper is organized as follows. Section \ref{method}  reviews variational Bayes and discusses its limitation in uncertainty quantification. Section~\ref{ss_boot} presents our proposed method utilizing fractional VB and a novel sample splitting-based calibration method that yields trustworthy VB. Section \ref{theory} provides theoretical guarantees that our method successfully calibrates the VB posterior, ensuring its credible intervals attain the nominal frequentist coverage. Section \ref{example} examines two canonical examples: the Bayesian Gaussian mixture model and the Bayesian mixture linear regression model, demonstrating their VB posterior approximations under the fractional posterior framework. In Section \ref{empirical}, we present a comparative study of three approaches: standard VB and two strategies of implementing our proposed trustworthy VB method, using the examples from Section \ref{example}. Section~\ref{real_data} provides a real-world data application. Section \ref{discussion} concludes the paper. Proofs are deferred to Section~\ref{proof}.

\section{Preliminaries on Variational Bayes}
 
\label{method}
 
This section reviews variational Bayes and its limitation in uncertainty quantification. We first introduce the model setup used throughout this article. Suppose we have $n$ observations $\bX = (X_{1}, \cdots, X_{n}) \overset{i.i.d.}{\sim} P_{\bth}$ with a $p$-dimensional unknown parameter $\bth = (\theta_{1}, \cdots, \theta_{p})^{\prime} \in \Theta$, where $\Theta \subset \mathbb{R}^{p}$ is the parameter space. Given the prior distribution $\pi(\bth)$ and likelihood function $p(\bX | \bth) = \prod\limits_{i = 1}^{n}p(X_{i} | \bth)$, the posterior distribution is obtained via $p(\bth | \bX) = \dfrac{\pi(\bth) p(\bX | \bth)}{\int_{\Theta}\pi(\bth) p(\bX | \bth) d\bth}$. Suppose we are interested in inference on $h(\bth)$, where $h(\cdot): \Theta \rightarrow \mathbb{R}$ is a continuous functional of $\bth$, such as a specific component or a transformation involving all components. In practice, users might be interested in one such $h(\bth)$ or multiple $h(\bth)$.

\subsection{Basic Variational Bayes}
\label{vb_meanfield}

When the posterior distribution $p(\bth | \bX)$ lacks tractability, inference on $h(\bth)$ often proceeds using sampling-based methods, such as MCMC \citep{hastings1970monte, gelfand1990sampling}. However, MCMC implementation can be computationally intensive, particularly when the model is complicated or high-dimensional\citep{hoffman2014no, betancourt2017conceptual}, motivating the need for more efficient methods of posterior approximation. Variational inference, or variational Bayes (VB), provides an alternative approach that is often computationally much more efficient to approximate the posterior distribution \citep{wainwright2008graphical}. 

VB seeks to find a function within a pre-specified space of distribution, denoted by $\mathcal{Q}$, that is closest to the true posterior distribution $p(\bth | \bX)$ under certain divergence, commonly the Kullback-Leibler (KL) divergence. Formally speaking, the variational approximation $q^{*}(\cdot)$ is obtained by
\begin{equation}\label{vb}
    q^{*} = \underset{q \in \mathcal{Q}}{\arg\min} \; \mathrm{KL}\left(q(\cdot) \parallel p(\cdot | \bX)\right)
    = \underset{q \in \mathcal{Q}}{\arg\min} \; \int_{\Theta} q(\bth) \log \left(\dfrac{q(\bth)}{p(\bth | \bX)}\right) d\bth. 
\end{equation}
In practice, VB turns to solving an equivalent optimization problem that maximizes the evidence lower bound (ELBO):
\begin{equation}\label{elbo}
    q^{*} = \underset{q \in \mathcal{Q}}{\arg\max} \; \underbrace{\int_{\Theta} q(\bth) \log \left(\dfrac{p(\bth, \bX)}{q(\bth)}\right) d\bth}_{\mathrm{ELBO}(q)},
\end{equation}
thereby bypassing the analytical intractability of the true posterior distribution \( p(\bth \mid \bX) \) in the original formulation~\eqref{vb}.

A commonly used approximation family for $\mathcal{Q}$ is the so-called mean-field variational family \citep{blei2017variational, wang2019frequentist, han2019statistical}, which assumes $\mathcal{Q} = \left\{q : q(\bth) = \prod\limits_{i = 1}^{p}q_{\theta_{i}}(\theta_{i})\right\}$. Other variational Bayes methods beyond the mean-field family are also available in the literature \citep{jaakkola1998improving, hoffman2015structured, tran2015copula, ranganath2016hierarchical}. For a comprehensive treatment of variational inference, we refer readers to an excellent review article \cite{blei2017variational}.

\subsection{VB's limitation in Uncertainty Quantification (UQ)}
\label{limit}

After obtaining the variational posterior approximation $q^{*}$, we can construct a $100(1 - \alpha)\%$ credible interval for a parameter of interest $h(\bth)$, denoted by $C_{q^{*}}(\bX;\alpha)$. However, a well-documented limitation of variational Bayes is its tendency to underestimate posterior covariance, producing approximations that are overly concentrated compared to the true posterior distribution $p(\cdot | \bX)$ \citep{wang2005inadequacy, bishop2006pattern, wang2019frequentist}. For example, under the mean-field approximations, the adopted independence assumption breaks posterior dependence structures. In particular, Theorem 5 and Corollary 9 in \cite{wang2019frequentist} show that the mean-field variational approximation $q^{*}$ converges to a Gaussian distribution but with a diagonal precision matrix, thereby discarding posterior correlations and underestimating posterior variances. 

This underestimation of the covariance matrix results in $100(1-\alpha)\%$ credible intervals constructed from $q^{*}$ that may not achieve the nominal frequentist coverage probability. 
Suppose $Q_{\bth}^{*}(\cdot)$ is the probability measure of $q^{*}(\cdot)$ and $P_{\bth}(\cdot)$ is the probability measure of $p(\cdot | \bX)$. We have $Q_{\bth}^{*}\left(\bth \in C_{q^{*}}(\bX;\alpha)\right) = 1 - \alpha$, but it might well be the case that $P_{\bth}\left(\bth \in C_{q^{*}}(\bX;\alpha)\right) < 1 - \alpha$ even in the limit. There is literature on correcting the variance estimate and improving the robustness of the traditional VB, such as \cite{giordano2018covariances}. However, statistical guarantees for interval estimation derived by VB require more than the corrected variance estimate. In contrast, under conditions, by Bernstein–von Mises theorem, the $100(1 - \alpha)\%$ credible interval constructed by the true posterior distribution will have correct frequentist coverage asymptotically \citep{van2000asymptotic, ghosal2017fundamentals}.

Therefore, with the computational advantage, VB sacrifices its trustworthiness in UQ that could be maintained by the true posterior distribution, in the sense of frequentist coverage. Undercoverage of $C_{q^{*}}(\bX;\alpha)$ is particularly concerning as it points to misleading statistical significance, and in this paper we will recalibrate it to the nominal frequentist coverage with theoretical guarantees, while maintaining a similar computational architecture of the VB.

\section{Trustworthy VB via Fractional Posterior Calibration}
\label{ss_boot}

To solve the undercoverage issue discussed in Section \ref{limit}, we propose a calibrating method based on fractional VB and sample splitting bootstrap. To begin with, recall that in classical Bayesian inference, we have $p(\bth | \bX) \propto \pi(\bth)p(\bX | \bth)$, where $p(\bX | \bth)$ is the likelihood function. For fractional posterior, we replace the likelihood function with $\exp\left(-\omega n R(\bth)\right)$, where $R(\bth) = -\dfrac{1}{n}\log p(\bX | \bth)$ is the negative log-likelihood function. This is a special case of Gibbs posterior where $R(\bth)$ is any empirical risk function \citep{alquier2016properties, syring2023gibbs, alquier2024user}. The corresponding posterior distribution is 
\begin{equation}
    p_{\omega}(\bth | \bX) = \dfrac{\pi(\bth)p^{\omega}(\bX | \bth)}{\int_{\Theta}\pi(\bth)p^{\omega}(\bX | \bth) d\bth};
\end{equation}
this posterior distribution is called fractional posterior and its theoretical properties have been studied in \cite{10.1214/18-AOS1712}.  Consider the more general case with latent variables $\bZ$, which is common in the applications of VB. Then the fractional posterior distribution is
\begin{equation}\label{gibbs_post}
    p_{\omega}(\bth, \bZ | \bX) = \dfrac{\pi(\bth)p^{\omega}(\bX | \bth, \bZ)p^{\omega}(\bZ | \bth)}{\int\pi(\bth)p^{\omega}(\bX | \bth, \bZ)p^{\omega}(\bZ | \bth) d\bth d\bZ}. 
\end{equation} 
The variational approximation $q_{\omega}^{*}(\bth)$ of the fractional posterior distribution is referred to as fractional VB \citep{yang2020alpha}, which extends tempered posterior or Gibbs posterior when there are latent variables.  With smaller $\omega \in (0, 1]$, the posterior distribution becomes more ``flattened'', which enlarges the credible interval and thus increases the frequentist coverage. For a special case when $\omega = 1$, the posterior distribution degenerates to the classical VB posterior distribution. 

Recall that $h(\cdot): \Theta \rightarrow \mathbb{R}$ is a continuous mapping of $\bth$ that defines our parameter of interest $h(\bth)$. Let $C_{q^{*}_{\omega}}(\bX; \alpha)$ be the $100(1 - \alpha)\%$ credible interval of $h(\bth)$ constructed by fractional variational posterior $q^{*}_{\omega}$, and $c_{q^{*}}(\omega; \alpha)$ be its frequentist coverage probability. The continuous fraction parameter $\omega$ induces credible intervals of varying width and, consequently, varying coverage. Our aim is to select an ``optimal'' $\omega$ such that the $100(1-\alpha)\%$ credible interval for $h(\bth)$ attains its nominal frequentist coverage $1-\alpha$. A discussion of related calibration literature can be found in Section~\ref{sec:rel_literature}. More specifically, we propose a sample splitting bootstrap-based procedure, called \textbf{t}rustworthy \textbf{v}ariational \textbf{B}ayes (TVB), that proceeds as follows: We partition the full dataset $\bX$ into $\bX_{1}$ and $\bX_{2}$, where $\bX_{1}$ serves to construct a surrogate for the true $h(\bth)$ at each $\omega \in (0, 1]$, and $\bX_{2}$ is used to generate bootstrap samples and estimate the empirical coverage for a given $\omega$. For a fixed $\omega$, we generate $B$ bootstrap datasets $\bX_{2}^{(1)}, \ldots, \bX_{2}^{(B)}$ by resampling from $\bX_{2}$ and compute the corresponding credible intervals $C_{q^{*}_{\omega}}(\bX_{2}^{(b)}; \alpha)$ for $b = 1, \ldots, B$, along with the posterior mean $\hat{\bth}_{\omega}^{\bX_{1}}$ obtained from the first split $\bX_{1}$. The frequentist coverage of $h(\bth)$ for a $100(1 - \alpha)\%$ credible interval under the $\omega$-VB procedure is then estimated via the bootstrap empirical coverage
\begin{equation}\label{eq:c^B_q}
    \hat{c}^{B}_{q^{*}}(\omega; \alpha) = \frac{\sum\limits_{b = 1}^{B}\mathbbm{1}\left\{h(\hat{\bth}_{\omega}^{\bX_{1}}) \in C_{q^{*}_{\omega}}(\bX_{2}^{(b)}; \alpha)\right\}}{B}.
\end{equation}
To achieve the nominal $100(1-\alpha)\%$ coverage, we would ideally select $\omega_{0}$ such that $\hat{c}^{B}_{q^{*}}(\omega; \alpha) = 1-\alpha$. In practice, we approximate this root by choosing
\begin{equation}\label{eq:optimize}
    \omega_{0} = \underset{\omega \in \Omega}{\arg\min} \left|\hat{c}^{B}_{q^{*}}(\omega; \alpha) - (1 - \alpha) \right|,
\end{equation}
where $\Omega \subseteq (0, 1]$ will be specified in the following Section~\ref{sec:seq_tvb} and Section~\ref{sec:grid_tvb} respectively. With the selected $\omega_{0}$, we use the corresponding fractional VB approximation to construct the $100(1 - \alpha)\%$ credible interval that attains the target coverage. 

Next we propose two approaches for optimizing $\omega_{0}$ in \eqref{eq:optimize}: a sequential updating method via stochastic approximation \citep{robbins1951stochastic}, and a grid-point search strategy that provides enhanced computational efficiency and flexibility.

\subsection{Strategy I: Sequential Update TVB}\label{sec:seq_tvb}
The sequential updating strategy is based on stochastic approximation: Starting with $\omega^{(0)} = 1$, and at $(k+1)$th iteration, we can update $\omega$ by \begin{equation}\label{update}
     \omega^{(k+1)} = \omega^{(k)} + \kappa_{k}\left(\hat{c}^{B}_{q^{*}}(\omega^{(k)}; \alpha) - (1 - \alpha)\right),
\end{equation}
where $\kappa_{k}$ is a user-defined step size satisfying $\sum \kappa_{k} = \infty, \sum \kappa_{k}^{2} < \infty$. The corresponding procedure is summarized in Algorithm \ref{ssb_cal}, with the use of reparametrization $\eta = \log\omega$ for numerical stability. For the sequential update of TVB, we assume $\Omega = [\delta, 1]$ for some small $\delta > 0$.
\begin{algorithm}[h!]
\caption{Trustworthy Variational Posterior Calibration with Sequential Update}\label{ssb_cal}
\begin{algorithmic}[1]

\REQUIRE Observation data $\bX$, target coverage $100(1 - \alpha)\%$, number of bootstrap samples $B$, mapping function $h(\cdot)$, step size $\kappa_{k}$, maximum number of iterations $\mathrm{max\_iter}$, tolerance $\epsilon > 0$. Initialization: $\eta^{(0)} = 0$, coverage = FALSE, $k = 0$

\WHILE{$k < \mathrm{max\_iter}$ \textbf{and} \textbf{not} coverage}
    \STATE Smoothing parameter $\omega^{(k)} = \exp(\eta^{(k)})$
    \STATE Randomly Split $\bX$ into $\bX_1$, $\bX_2$ with the same number of samples
    \STATE Use $\bX_{1}$ to obtain variational approximation $q^{*}_{\omega^{(k)}}(\bth)$ of fractional posterior \eqref{gibbs_post} with corresponding posterior mean $\hat{\bth}_{\omega^{(k)}}^{\bX_{1}}$
    \STATE Implement non-parametric bootstrap on $\bX_{2}$ and get $B$ bootstrap samples $\bX_{2}^{(1)}, \cdots, \bX_{2}^{(B)}$
    
    \FOR{$b = 1$ \TO $B$}
        \STATE Obtain mean-field variation approximation $q^{(b)}_{\omega^{(k)}}(\bth)$ using $\bX_{2}^{(b)}$
        \STATE Construct $100(1 - \alpha)\%$ credible interval $C_{q^{(b)}_{\omega^{(k)}}}(\bX_{2}^{(b)}; \alpha)$
    \ENDFOR
    
    \STATE Obtain empirical bootstrap coverage probability $\hat{c}^{B}_{q^{*}}(\omega^{(k)}; \alpha)$ by
    \begin{equation}\label{boot_coverage}
        \hat{c}^{B}_{q^{*}}(\omega^{(k)}; \alpha) = \dfrac{\sum\limits_{b = 1}^{B}\mathbbm{1}\left\{h(\hat{\bth}_{\omega^{(k)}}^{\bX_{1}}) \in C_{q^{(b)}_{\omega^{(k)}}}(\bX_{2}^{(b)}; \alpha)\right\}}{B}
    \end{equation}
    
    \IF{$\hat{c}^{B}_{q^{*}}(\omega^{(k)}; \alpha) > 1 - \alpha$ or $(1 - \alpha) - \hat{c}^{B}_{q^{*}}(\omega^{(k)}; \alpha) < \epsilon$}
        \STATE coverage = TRUE
    \ELSE
        \STATE Update $\eta$ by $\eta^{(k+1)} = \eta^{(k)} + \kappa_{k}\left(\hat{c}^{B}_{q^{*}}(\omega^{(k)}; \alpha) - (1 - \alpha)\right)$
        \STATE $k = k + 1$
    \ENDIF
\ENDWHILE

\STATE Return $\hat{\omega}_{0} = \omega^{(k)}$
\end{algorithmic}
\end{algorithm}

\subsection{Strategy II: Grid-Search TVB}\label{sec:grid_tvb}
The sequential updating TVB approach discussed in Section \ref{sec:seq_tvb} presents two significant limitations. First, the method's iterative nature, requiring multiple runs of the variational Bayes algorithm at each step, results in substantial computational overhead when determining $\omega_{0}$. Second, and perhaps more critically, the selected $\omega_{0}$ is inherently ``parameter-specific''; that is, for any new component or transformation of $\bth$ (such as $h_1(\bth)$), Algorithm \ref{ssb_cal} must be executed anew.

We address both limitations through a grid-search strategy. Consider a grid of values $\boldsymbol{\omega} = (\omega_{1}, \cdots, \omega_{m})^{\prime}$, where $m$ denotes the number of grid points satisfying $0 < \omega_{1} < \cdots < \omega_{m} < 1$. Our procedure has two phases: table/dictionary construction and table/dictionary utilization. In the first phase, we run the VB algorithm in parallel for each $\omega_{i}$, $i \in [m]$, and store the resulting variational posteriors in a dictionary-like structure. This ``TVB table'' serves as a reusable reference for subsequent queries. In the second phase, for any given functional $h(\bth)$, we evaluate the criterion in \eqref{eq:optimize} over the grid and select the $\omega$ that minimizes it. The resulting implementation is summarized in Algorithm~\ref{ssb_grid}. For grid-search TVB, the parameter space $\Omega$ in equation \eqref{eq:optimize} corresponds to $\boldsymbol{\omega}$.

The idea of precomputing multiple, parallel runs of an algorithm over a grid and packing the results into a table that enables efficient, parameter-agnostic look-up is reminiscent of cycle spinning in multiscale image analysis, including wavelet-based methods \citep{li2015fast, coifman1995translation}, where a translation-invariant table is constructed to encode information across all shifts. This analogy motivates our terminology of a “TVB table” for the grid-search implementation.

\begin{algorithm}[h!]
\caption{Trustworthy Variational Posterior Calibration with Grid-Search Update}\label{ssb_grid}
\begin{algorithmic}[1]

\REQUIRE Observation data $\bX$, target coverage $100(1 - \alpha)\%$, number of bootstrap samples $B$, mapping function $h(\cdot)$, grid points $\boldsymbol{\omega} = (\omega_{1}, \cdots, \omega_{m})^{\prime}$

\STATE \textbf{Dictionary Construction:}

\FOR{$k = 1, 2, \cdots, m$}
    \STATE Randomly Split $\bX$ into $\bX_1$, $\bX_2$ with the same number of samples
    \STATE Use $\bX_{1}$ to obtain variational approximation $q^{*}_{\omega_{k}}(\bth)$ of fractional posterior \eqref{gibbs_post} with corresponding posterior mean $\hat{\bth}_{\omega_{k}}^{\bX_{1}}$
    \STATE Implement non-parametric bootstrap on $\bX_{2}$ and get $B$ bootstrap samples $\bX_{2}^{(1)}, \cdots, \bX_{2}^{(B)}$
    
    \FOR{$b = 1$ \TO $B$}
        \STATE Obtain mean-field variation approximation $q^{(b)}_{\omega_{k}}(\bth)$ using $\bX_{2}^{(b)}$
    \ENDFOR

    \STATE Save the VB estimation with bootstrap estimation w.r.t. $\omega_{k}$ in a dictionary $\mathfrak{D}$
\ENDFOR

\STATE \textbf{Dictionary Utilization:}

\FOR{$k = 1, 2, \cdots, m$}
    \FOR{$b = 1$ \TO $B$}
        \STATE Using $\mathfrak{D}$, construct $100(1 - \alpha)\%$ credible interval $C_{q^{(b)}_{\omega_{k}}}(\bX_{2}^{(b)}; \alpha)$ of $h(\bth)$
        \STATE Obtain empirical bootstrap coverage probability $\hat{c}^{B}_{q^{*}}(\omega_{k}; \alpha)$ by
        \begin{equation}
            \hat{c}^{B}_{q^{*}}(\omega_{k}; \alpha) = \dfrac{\sum\limits_{b = 1}^{B}\mathbbm{1}\left\{h(\hat{\bth}_{\omega_{k}}^{\bX_{1}}) \in C_{q^{(b)}_{\omega_{k}}}(\bX_{2}^{(b)}; \alpha)\right\}}{B}
        \end{equation}
    \ENDFOR

    \STATE Compare all $\hat{c}^{B}_{q^{*}}(\omega_{k}; \alpha)$, find $k^{\prime}$ such that \eqref{eq:optimize} is minimized
\ENDFOR

\STATE Return $\hat{\omega}_{0} = \omega_{k^{\prime}}$
\end{algorithmic}
\end{algorithm}

Our numerical experiments in Section \ref{empirical} demonstrate that both strategies achieve effective coverage calibration. Nevertheless, we advocate for the grid-search approach because it is parameter-agnostic, applies broadly, and readily supports parallel computation.

\subsection{Related Literature}
\label{sec:rel_literature}
Although not specifically focused on VB, \cite{syring2019calibrating} proposes the general posterior calibration (GPC) algorithm to update the tuning parameter $\omega$ in the Gibbs posterior, by comparing the bootstrap empirical coverage with the target $100(1 - \alpha)\%$ coverage. We extend the GPC algorithm in the context of VB on two fronts. First, we build on the recent fractional VB method developed in \citep{yang2020alpha}, which differs from the Gibbs posterior considered in \cite{syring2019calibrating} particularly in the presence of latent variables. We have found that fractional VB is critical to ensure the validity of recalibrated credible intervals. Also fractional VB has appealing theoretical properties, including its posterior consistency for any $\omega \in (0, 1]$, which ensures certain assumptions required to establish our theory (see the next section and remark after Assumption \ref{assump_consistency}). Second, since GPC treats the empirical bootstrap coverage $\hat{c}^{B}_{q^{*}}(\omega; \alpha)$ as an approximation of the true frequentist coverage $c_{q^{*}}(\omega; \alpha)$, when using the full dataset $\bX$ to implement both original estimation and bootstrap estimation, the corresponding $\hat{c}^{B}_{q^{*}}(\omega; \alpha)$ may suffer from ``overconfidence'' as $h(\hat{\bth}_{\omega})$ (where $\hat{\bth}_{\omega}$ denotes the $\omega$-VB posterior mean computed on the full dataset $\bX$, which differs from $\hat{\bth}_{\omega}^{\bX_{1}}$) is not independent of $C_{q^{*}_{\omega}}(\bX^{(b)}; \alpha)$. Specifically, the optimal $\omega$ obtained through the GPC algorithm may not be sufficiently small to achieve proper coverage probability calibration for the credible set $C_{q^{*}_{\omega}}(\bX;\alpha)$. Motivated by split conformal prediction \citep{lei2018distribution} and bootstrap with split sample \citep{10.1214/18-AOS1784}, we propose the sample splitting bootstrap procedure integrated with GPC algorithm to improve the selection of $\omega$, with corresponding theoretical guarantees. 

Existing procedures, such as that of~\cite{syring2019calibrating}, lack formal guarantees for the resulting uncertainty quantification, which is particularly concerning because coverage assessment is typically more demanding than assessing point estimation accuracy. In contrast, as we show later, our theoretical results both establish asymptotically correct frequentist coverage for calibrated fractional VB credible intervals and clarify the conditions required of the underlying VB approximation. Informally, they indicate how far one can ``bend'' a VB posterior, and which variational approximations are amenable to such ``mending’’, which carries direct practical relevance.

\section{Theoretical Results}
\label{theory}

In this section we establish theoretical results to verify the validity of the proposed method, i.e., the selected $\hat{\omega}_{0}$ in~\eqref{eq:optimize} can recalibrate the credible interval for the corresponding variational posterior distribution, such that $c_{q^{*}}(\hat{\omega}_{0}; \alpha)$ converges to $1 - \alpha$. Recall that the parameter of interest we wish to recalibrate credible intervals for is given by a general continuous mapping $h(\cdot): \Theta \rightarrow \mathbb{R}$.

We will first show that the bootstrap coverage probability $\hat{c}_{q^{*}}^{B}(\omega; \alpha)$ converges in probability to the true coverage probability $c_{q^{*}}(\omega; \alpha)$ for all $\omega$. Our theoretical results are based on empirical process theory, especially the conditional multiplier central limit theorem and functional delta-method~\citep{wellner2013weak}. The following conditions are assumed: 

\begin{assump}[Consistency]
\label{assump_consistency}
    Let $\hat{\bth}_{\omega} = \int_{\Theta}\bth q^{*}_{\omega}(\bth)d\bth$ be the variational approximation posterior mean under smoothing parameter $\omega$. The plug-in estimator $h(\hat{\bth}_{\omega})$ is consistent, i.e., $h(\hat{\bth}_{\omega}) \overset{P}{\rightarrow} h(\bth)$ for all $\omega \in (0, 1]$ as $n \rightarrow \infty$.
\end{assump}

\begin{rmk}
\label{rmk_consistency}
    Assumption \ref{assump_consistency} requires that the plug-in posterior mean estimator under fractional VB is consistent for every $\omega \in (0, 1]$. This is a mild condition, as the existing theory for fractional VB already establishes stronger results. For example, for $\omega \in (0, 1]$, \cite{yang2020alpha} showed that the variational posterior mean $\hat{\bth}_{\omega}$ achieves minimax concentration rates around the true parameter $\bth$, which is stronger than consistency. When $h(\cdot)$ is a continuous function, the continuous mapping theorem ensures that $h(\hat{\bth}_{\omega}) \overset{P}{\rightarrow} h(\bth)$. 
\end{rmk}

   Let $f_{\omega}(X_{i}) = p^{\omega}(X_{i} | \bth, \bZ)p^{\omega}(\bZ | \bth)$ as in \eqref{gibbs_post}, which also includes the simplified case $f_{\omega}(X_{i}) = p^{\omega}(X_{i} | \bth)$ when there is no latent variable $\bZ$. 
\begin{assump}[Donsker Class]
\label{Donsker}
 Suppose $\mathcal{F} := \left\{f_{w}(\cdot): \omega \in (0, 1]\right\}$ is a Donsker class and $f_{\omega}$ is measurable for all $\omega \in (0, 1]$. In addition, for any $\delta > 0$, the class $\mathcal{F}_{\delta} := \left\{f_{\omega_{1}} - f_{\omega_{2}}: f_{\omega_{1}} \in \mathcal{F}, f_{\omega_{2}} \in \mathcal{F} \text{ and } \rho_{P}(f_{\omega_{1}} - f_{\omega_{2}}) < \delta\right\}$ is measurable, where $\rho_{P}(\cdot)$ is any metric that depends on true data distribution $P$.
\end{assump}

\begin{rmk}
\label{rmk_donsker}
     Assumption \ref{Donsker} is useful for establishing the uniform convergence of bootstrap-based credible interval coverage to that of the original sample. Specifically, since different values of $\omega$ generate distinct posterior distributions, this assumption ensures that the convergence holds uniformly across the entire range of $\omega$ values. This assumption can be verified using a standard Donsker theorem argument (see, for example, \cite{sen2018gentle}). Suppose there is an envelope function $F$ of $\mathcal{F}$, i.e. $\underset{f_{\omega} \in \mathcal{F}}{\sup} |f_{\omega}(x)| < F(x)$ for all $x$, such that $\mathbb{E}_{P}[F] < \infty$, and that the bracketing entropy satisfies $J_{[]}\left(1, \mathcal{F}, L_{2}(P)\right):= \int_{0}^{1}\sqrt{\log N_{[]}\left(\epsilon, \mathcal{F} \cup \{0\}, L_{2}(P)\right)} d\epsilon < \infty$. Then $\mathcal{F}$ is a Donsker class. When there is no latent variable, $\mathcal{F}$ is the space of all power likelihood functions $f_{\omega}(X_{i}) = p^{\omega}(X_{i} | \bth)$ with $\omega \in (0, 1]$. For many common models (e.g., Gaussian and other exponential family distributions), the distribution is upper-bounded by a constant, so there exists some constant $M < \infty$ with $\underset{f_{\omega} \in \mathcal{F}}{\sup} |f_{\omega}(x)| < M$, and we can let $M$ be the envelope function, which satisfies $\mathbb{E}_{P}[M] < \infty$. For the bracketing condition, suppose $f_{\omega}(x) \in \mathcal{F}$ is Lipschitz in $\omega$ under the $L_{2}(P)$ metric with Lipschitz constant $L < \infty$, that is, $\|f_{\omega_{1}}-f_{\omega_{2}}\|_{L_{2}(P)}\le L|\omega_{1}-\omega_{2}|.$
     Then for any $\epsilon > 0$, we can cover $(0,1]$ by at most $L/\epsilon$ grid points in $\omega$, which implies $\log N_{[]}\left(\epsilon, \mathcal{F} \cup \{0\}, L_{2}(P)\right) = O(\log(\epsilon^{-1}))$, and hence $J_{[]}\left(1, \mathcal{F}, L_{2}(P)\right) < \infty$. The extension of this argument to the case with latent variables follows analogously under similar suitable conditions.
\end{rmk}

Our next assumption imposes a regularity on credible intervals, formulated via Hadamard differentiability (see, \cite{wellner2013weak}, for example). A map $T: \mathbb{D} \rightarrow \mathbb{E}$ is Hadamard differentiable at $\xi \in \mathbb{D}$ if there is a continuous linear map $T^{\prime}_{\xi}: \mathbb{D} \rightarrow \mathbb{E}$ that satisfies $\lim\limits_{a \rightarrow 0, n \rightarrow \infty}\dfrac{T(\xi + ab_{n}) - T(\xi)}{a} = T_{\xi}^{\prime}(b)$ for sequence $a \in \mathbb{R}, b_{n} \rightarrow b \in \mathbb{D}$. 

\begin{assump}[Hadamard Differentiability]
\label{hadamard}
    For a fixed $\bth$, let $\mathbb{P}_{n}$ denote the empirical measure of the observation data $\bX = (X_{1}, \cdots, X_{n})$, and let $P_{\bth}$ denote the true data-generating distribution. We define the mapping $T_{\bth}(\cdot)$ that takes an empirical measure $\mathbb{P}_{n}$ and produces a functional $T_{\bth}(\mathbb{P}_{n}): \mathcal{F} \to \mathbb{R}$ given by
    \[
    T_{\bth}(\mathbb{P}_{n})(f_{\omega}) = \sigma\left(h(\bth) - L(\mathbb{P}_{n})(f_{\omega})\right) \cdot \sigma\left(U(\mathbb{P}_{n})(f_{\omega}) - h(\bth)\right),
    \]
    where $L(\mathbb{P}_{n}), U(\mathbb{P}_{n}): \mathcal{F} \to \mathbb{R}$ are functionals that extract the lower and upper bounds of the credible interval from any fractional posterior $f_{\omega} \in \mathcal{F}$. Here, $\sigma(z) = (1 + e^{-kz})^{-1}$ is the sigmoid function that provides a smooth approximation to the coverage indicator. We assume that $T_{\bth}(\cdot)$ is Hadamard differentiable at $P_{\bth}$.
\end{assump}

\begin{rmk}
\label{rmk_hadamard}
    Hadamard differentiability describes how smoothly a mapping $T_{\bth}(\mathbb{P}_{n})$ changes under a small perturbation direction $b_{n} \rightarrow b$. If it behaves almost linearly in response to the perturbation, then the map is Hadamard differentiable. 
    For any $\bth$ and $\mathbb{P}_{n}$, the operator $T_{\bth}(\mathbb{P}_{n})$ in Assumption \ref{hadamard} maps the class $\mathcal{F}$ to a real value between $(0, 1)$. For a specific $f_{\omega} \in \mathcal{F}$, $T_{\bth}(\mathbb{P}_{n})(f_{\omega})$ is a smoothing approximation of indicator function $\mathbbm{1}\{h(\bth) \in [L(\mathbb{P}_{n})(f_{\omega}), U(\mathbb{P}_{n})(f_{\omega})]\}$. The sigmoid function used in the assumption is a smooth surrogate of the indicator function. To apply functional delta-method, it is necessary to make sure that $T_{\bth}(\cdot)$ is Hadamard differentiable, which requires both $L(\cdot)$ and $U(\cdot)$ to be Hadamard differentiable. For more detail of functional delta-method and Hadamard differentiability, readers are referred to \cite{van1991differentiable, van1991efficiency}.
\end{rmk}

\begin{thm}[Convergence of Coverage Probability]
\label{converge}
    Suppose $c_{q^{*}}(\omega; \alpha)$ is the frequentist coverage probability of the $100(1 - \alpha)\%$ credible interval of $h(\bth)$ constructed from the fractional VB posterior distribution $q^{*}_{\omega}$, and $\hat{c}_{q^{*}}^{B}(\omega; \alpha)$ is the bootstrap empirical coverage probability defined in \eqref{eq:c^B_q}. 
    If Assumptions \ref{assump_consistency}-\ref{hadamard} hold, then for each fixed $\omega \in (0, 1]$, 
    \begin{equation}\label{thm_result}
        \hat{c}_{q^{*}}^{B}(\omega; \alpha) \overset{P}{\rightarrow} c_{q^{*}}(\omega; \alpha),
    \end{equation}
    as $n, B \rightarrow \infty$. 
\end{thm}

The proof of Theorem \ref{converge} is provided in Section \ref{proof}.

\begin{rmk}
\label{rmk_thm}
    Theorem \ref{converge} claims that under conditions, the empirical coverage probability of split bootstrap samples will converge to that of the original samples for any $\omega \in (0, 1]$. Therefore, if we can select $\omega_{0}$ such that $\hat{c}_{q^{*}}^{B}(\omega_{0}; \alpha) = 100(1 - \alpha)\%$, the corresponding credible interval constructed by original data $\bX$ with $\omega_{0}$ can also attain target frequentist coverage probability. This is formally described in Theorem \ref{cor_coverage} under additional mild assumptions.
\end{rmk}

Theorem~\ref{converge} ensures that the bootstrap empirical coverage $\hat{c}_{q^{*}}^{B}(\omega; \alpha)$ converges to the true coverage $c_{q^{*}}(\omega; \alpha)$ pointwise for all $\omega \in (0, 1]$. To ensure that the coverage equipped with the $M$-estimator obtained by solving \eqref{eq:optimize} converges to its nominal level, we also require uniform convergence over $\Omega$. This follows straightforwardly from the proof of Theorem~\ref{converge} under an additional mild assumption on the fractional VB posterior mean.

\begin{assump}\label{assump:uniform}
    For the fractional VB posterior mean $\hat{\bth}_{\omega}$ defined above, assume $h(\hat{\bth}_{\omega}) \overset{P}{\rightarrow} h(\bth)$ uniformly for all $\omega \in \Omega$ (i.e., $\underset{\omega \in \Omega}{\sup}\|h(\hat{\bth}_{\omega}) - h(\bth) \| \overset{P}{\rightarrow} 0$) as $n \rightarrow \infty$.
\end{assump}

\begin{rmk}\label{rmk:uniform_ass}
    Assumption~\ref{assump:uniform} strengthens Assumption~\ref{assump_consistency} by requiring uniform convergence over $\Omega$. When $\Omega = \boldsymbol{\omega}$ in grid-search TVB, uniform convergence on the finite grid is automatically satisfied since we have pointwise convergence. For the continuous case $\Omega = [\delta, 1]$ for small $\delta > 0$, \cite{yang2020alpha} establish uniform convergence on $[\delta, 1 - \delta]$. While their risk bound diverges as $\omega \rightarrow 1$, consistency at the boundary $\omega = 1$ holds under stronger regularity conditions \citep{yang2020alpha}.  Therefore, uniform convergence of $\hat{\bth}_{\omega}$ can be established under sufficiently strong conditions when $\Omega = [\delta, 1]$. By the continuous mapping theorem, uniform convergence of $h(\hat{\bth}_{\omega})$ follows immediately.
\end{rmk}

Under Assumption~\ref{assump:uniform}, we can extend the pointwise convergence results in Theorem~\ref{converge} to uniform convergence over $\Omega$, as stated below.

\begin{cor}[Uniform Convergence of Coverage Probability]\label{cor:uniform}
    Suppose $c_{q^{*}}(\omega; \alpha)$ is the frequentist coverage probability of the $100(1 - \alpha)\%$ credible interval of $h(\bth)$ constructed from the fractional VB posterior distribution $q^{*}_{\omega}$, and $\hat{c}_{q^{*}}^{B}(\omega; \alpha)$ is the bootstrap empirical coverage probability defined in \eqref{eq:c^B_q}. 
    If Assumptions \ref{Donsker}-\ref{assump:uniform} hold, then 
    \begin{equation}\label{thm_result}
        \underset{\omega \in \Omega}{\sup}\left|\hat{c}_{q^{*}}^{B}(\omega; \alpha) - c_{q^{*}}(\omega; \alpha)\right| \overset{P}{\rightarrow} 0,
    \end{equation}
    as $n, B \rightarrow \infty$. 
\end{cor}

The proof of Corollary~\ref{cor:uniform} follows as a straightforward extension of the proof of Theorem~\ref{converge} and is provided in Section~\ref{proof}. With Corollary~\ref{cor:uniform} established, to ensure that our calibration procedure is well-defined and that the $M$-estimator from \eqref{eq:optimize} converges to a unique true value of the smoothing parameter, we introduce the following assumption.

\begin{assump}
\label{assump:uniqueness}
For any target coverage level $1 - \alpha \in (0, 1)$, there exists a unique $\omega^{*} \in \Omega$, such that the $100(1 - \alpha)\%$ credible interval from  $\omega^{*}$-VB posterior achieves the nominal coverage, namely $$\lim_{n \rightarrow \infty} \mathbb{P}\left(h(\bth) \in C_{q^{*}_{\omega^{*}}}(\bX; \alpha) \right) = 1 - \alpha.$$ Additionally, assume $\mathbb{P}\left(h(\bth) \in C_{q^{*}_{\omega^{*}}}(\bX; \alpha) \right)$ is equicontinuous as function of $\omega \in (0, 1]$ for all $n$.
\end{assump}

\begin{thm}
\label{cor_coverage}
    Let $\hat{\omega}_{0}$ be the smoothing parameter estimated by Algorithm \ref{ssb_cal} or Algorithm \ref{ssb_grid}. For continuous map $h(\cdot): \Theta \rightarrow \mathbb{R}$, let $C_{q^{*}_{\hat{\omega}_{0}}}(\bX; \alpha)$ be the $100(1 - \alpha)\%$ credible interval of $h(\bth)$ based on variational posterior $q^{*}_{\hat{\omega}_{0}}$. Under Assumptions \ref{Donsker}-\ref{assump:uniqueness}, we have
    \begin{equation}\label{main_cor}
        \mathbb{P}\left(h(\bth) \in C_{q^{*}_{\hat{\omega}_{0}}}(\bX; \alpha)\right) \overset{P}{\rightarrow} 1 - \alpha,
    \end{equation}
    as $n, B \rightarrow \infty$.
\end{thm}
Theorem~\ref{cor_coverage} provides a theoretical guarantee that the $M$-estimator selected by Algorithm~\ref{ssb_cal} or Algorithm~\ref{ssb_grid} achieves nominal $100(1 - \alpha)\%$ coverage. We verify this result through numerical analysis across different scenarios in the following sections. The proof of Theorem~\ref{cor_coverage} is provided in Section~\ref{proof}.

\section{Two examples: Gaussian Mixture Model and Bayesian Mixture Linear Regression Model}
\label{example}
In this section, we use Gaussian Mixture Model and Bayesian Mixture Linear Regression Model to illustrate fractional VB that our proposed TVB method relies on.

We use the mean-field variational family for $\mathcal{Q}$ in both examples. The mean-field variational family is a factorizable distribution family given by $\mathcal{Q} = \left\{q : q(\bth) = \prod\limits_{i = 1}^{p}q_{\theta_{i}}(\theta_{i})\right\}$. To solve optimization problem \eqref{elbo} with $\mathcal{Q}$ being mean-field variational family, we apply a coordinate ascent algorithm (CAVI) by updating $q_{\theta_{i}}(\cdot)$ coordinately at each iteration until convergence. In particular, at each iteration, $q^{*}_{\theta_{i}}(\cdot)$ is obtained by
\begin{equation}\label{CAVI}
    \log q^{*}_{\theta_{i}} = \int_{\Theta} \left(p(\bth, \bX) \prod_{j \neq i} q^{*}_{\theta_{j}}(\theta_{j}) \right)d\theta_{i}\cdots d\theta_{i - 1} d\theta_{i} \cdots d\theta_{p} + C,
\end{equation}
where $C$ is a normalizing constant for the distribution $q^{*}_{\theta_{i}}$. Once the algorithm converges, we can use $q^{*}(\bth) = \prod\limits_{i = 1}^{p} q_{\theta_{i}}^{*}(\theta_{i})$ as the mean-field variational approximation of the posterior distribution $p(\bth | \bX)$.

However, the mean-field variational updates for the fractional posterior with latent variables in \eqref{gibbs_post} require specific modifications, which we detail in the algorithmic implementations for each example below.

\subsection{Gaussian Mixture Model}
\label{gmm}
Suppose each row of the observation matrix $\bX \in \mathbb{R}^{N \times p}$ is a $p$-dimensional random vector drawn from Gaussian mixture model (GMM) with $K$ clusters, such that $\bx_{n} \sim \sum\limits_{k = 1}^{K} \pi_{k} N(\bmu_{k}, \bLam_{k}^{-1})$, where the mixing coefficient of $K$ clusters is $\bpi = (\pi_{1}, \cdots, \pi_{K})^{\prime}$. The membership matrix $\bZ \in \mathbb{R}^{N \times K}$ consists of binary entries $z_{nk}$, where $z_{nk} = 1$ indicates that the $n$th observation $\bx_{n}$ belongs to cluster $k$. For each cluster, we denote the mean vectors as $\bmu = (\bmu_{1}, \cdots, \bmu_{K})^{\prime}$ and the precision matrices $\bLam = (\bLam_{1}, \cdots, \bLam_{K})^{\prime}$. The prior distribution and likelihood function are
\begin{equation}\label{prior_gmm}
\begin{aligned}
    & p\left(\bpi\right) = \mathrm{Dir}(\bpi | \bal_{0}) = C(\bal_{0})\prod_{k = 1}^{K}\pi_{k}^{\alpha_{0} - 1}\\
    & p\left(\bZ | \bpi\right) = \prod_{n = 1}^{N}\prod_{k = 1}^{K}\pi_{k}^{z_{nk}}\\
    & p\left(\bmu, \bLam\right) = p\left(\bmu | \bLam\right)p\left(\bLam\right) = \prod_{k = 1}^{K}N\left(\bmu_{k} | \bm_{0}, (\beta_{0}\bLam_{k})^{-1}\right)\mathcal{W}\left(\bLam_{k} | \bW_{0}, \nu_{0}\right)\\
    & p\left(\bX | \bZ, \bmu, \bLam\right) = \prod_{n = 1}^{N}\prod_{k = 1}^{K}N\left(\bx_{n} | \bmu_{k}, \bLam_{k}^{-1}\right)^{z_{nk}},
\end{aligned}
\end{equation}
where $\bal_{0} = (\alpha_{0}, \cdots, \alpha_{0})^{\prime}$, $\bm_{0}, \beta_{0}, \bW_{0}$ and $\nu_{0}$ are hyperparameters that will be updated. This is a latent variable model, and we have both the global parameters $\left(\bpi, \bmu, \bLam\right)$ and local parameter (latent variable) $\bZ$. Let $\bth = (\bpi, \bZ, \bmu, \bLam)$. We use the following fractional posterior distribution \citep{han2019statistical, yang2020alpha}: 
\begin{equation}\label{joint_gmm}
    p\left(\bX, \bpi, \bZ, \bmu, \bLam\right) = p^{\omega}\left(\bX | \bZ, \bmu, \bLam\right)p^{\omega}\left(\bZ | \bpi\right)p\left(\bpi\right)p\left(\bmu | \bLam\right)p\left(\bLam\right).
\end{equation}
The mean-field variational approximation is obtained by maximizing the following modified ELBO by assuming the factorization that $q\left(\bpi, \bZ, \bmu, \bLam\right) = q\left(\bZ\right)q\left(\bpi, \bmu, \bLam\right)$
\begin{equation}\label{elbo_gmm}
    q^{*} = \underset{q \in \mathcal{Q}}{\arg\max} \; \int_{\Theta} q(\bpi, \bmu, \bLam)q(\bZ) \log \left(\dfrac{p\left(\bX, \bpi, \bZ, \bmu, \bLam\right)}{q(\bpi, \bmu, \bLam)q^{\omega}(\bZ)}\right) d\bth.
\end{equation}
For notational simplicity, we suppress the subscripts of each $q(\cdot)$ in equation \eqref{elbo_gmm}. Algorithm \ref{CAVI_gmm} summarizes the procedure for solving optimization problem \eqref{elbo_gmm}. The detailed derivation and evaluation of the corresponding ELBO are provided in the supplemental materials. 

\begin{algorithm}[h!]
\caption{Fractional Mean-Field Variational Inference for GMM}\label{CAVI_gmm}
\begin{algorithmic}[1]

\REQUIRE Observation data $\bX \in \mathbb{R}^{N \times p}$, number of clusters $K$, hyperparameters $\bal_{0} = (\alpha_{0}, \cdots, \alpha_{0})^{\prime}$, $\bm_{0}, \beta_{0}, \bW_{0}, \nu_{0}$, smoothing parameter $\omega$, maximum number of iterations $\mathrm{max\_iter}$, tolerance $\epsilon > 0$. Initialization: coverage = FALSE, posterior parameters $\left(\alpha_{k}^{(0)}, \bm_{k}^{(0)}, \beta_{k}^{(0)}, \nu_{k}^{(0)}, \bW_{k}^{(0)}\right)$ for $k = 1, \cdots, K$, iteration number $t = 0$.

\WHILE{$t < \mathrm{max\_iter}$ \textbf{and} \textbf{not} coverage}
    \STATE Update posterior parameters of $q^{*}(\bZ)$ by
    \begin{equation}\label{update_qz_gmm}
        r_{nk} = \dfrac{\rho_{nk}}{\sum_{j = 1}^{K}\rho_{nj}},
    \end{equation}
    where $\log \rho_{nk} = \mathbb{E}_{\bpi}\left[\log \pi_{k}\right] + \dfrac{1}{2}\mathbb{E}_{\bLam}\left[\log |\bLam_{k}|\right] - \dfrac{1}{2}\mathbb{E}_{\bmu, \bLam}\left[(\bx_{n} - \bmu_{k})^{\prime}\bLam_{k}(\bx_{n} - \bmu_{k})\right]$

    \STATE For $k = 1, \cdots, K$, update posterior parameters of $q^{*}(\pi_{k})$ by
    \begin{equation}\label{update_qpi_gmm}
        \alpha_{k} = \alpha_{0} + \omega \cdot N_{k},
    \end{equation}
    where $N_{k} = \sum\limits_{n = 1}^{N} r_{nk}$

    \STATE For $k = 1, \cdots, K$, update posterior parameters of $q^{*}(\bmu_{k} | \bLam_{k})$ by
    \begin{equation}\label{update_qmu_gmm}
        \beta_{k} = \beta_{0} + \omega \cdot N_{k}, \quad
        \bm_{k} = \dfrac{1}{\beta_{k}}\left(\beta_{0}\bm_{0} + \omega \cdot N_{k}\bar{\bx}_{k}\right),
    \end{equation}
    where $\bar{\bx}_{k} = \dfrac{1}{N_{k}}\sum\limits_{n = 1}^{N} r_{nk}\bx_{n}$

    \STATE For $k = 1, \cdots, K$, update posterior parameters of $q^{*}(\bLam_{k})$ by
    \begin{equation}\label{update_qLam_gmm}
    \begin{aligned}
        & \nu_{k} = \nu_{0} + \omega \cdot N_{k}\\
        & \bW_{k}^{-1} = \bW_{0}^{-1} + \omega \cdot N_{k}\bS_{k} + \dfrac{\beta_{0}N_{k}}{\beta_{0} + N_{k}}(\bar{\bx}_{k} - \bm_{0})^{\prime}(\bar{\bx}_{k} - \bm_{0}),
    \end{aligned}
    \end{equation}
    where $\bS_{k} = \dfrac{1}{N_{k}}\sum\limits_{n = 1}^{N}r_{nk}(\bx_{n} - \bar{\bx}_{k})(\bx_{n} - \bar{\bx}_{k})^{\prime}$
\ENDWHILE

\STATE Return $q^{*}(\bZ) = \prod\limits_{n = 1}^{N}\prod\limits_{k = 1}^{K} r_{nk}^{z_{nk}}$, $q^{*}(\bpi) = \mathrm{Dir}(\bpi | \bal)$ where $\alpha_{k}$ is in \eqref{update_qpi_gmm}, $q^{*}(\bmu|\bLam) = \prod\limits_{k = 1}^{K} q^{*}(\bmu_{k}|\bLam_{k}) = \prod\limits_{k = 1}^{K} N\left(\bmu_{k} | \bm_{k}, (\beta_{k}\bLam_{k})^{-1}\right)$ and $q^{*}(\bLam) = \prod\limits_{k = 1}^{K} q^{*}(\bLam_{k}) = \prod\limits_{k = 1}^{K} \mathcal{W}(\bLam_{k} | \bW_{k}, \nu_{k})$.
\end{algorithmic}
\end{algorithm}

It is noteworthy that Algorithm \ref{CAVI_gmm} has the similar updating rule as the traditional variational approximation (for example, see \cite{bishop2006pattern}), except for the presence of the parameter $\omega$ in most terms.
 
Section \ref{empirical} contains numerical results for our proposed TVB method under this model. To apply the grid-search TVB method with $N$ samples, $K$ clusters, and $p$-dimensional data on a grid $\boldsymbol{\omega} \in (0, 1]^{m}$ with $B$ bootstrap replicates, we require a dictionary $\mathfrak{D} \in \mathbb{R}^{\mathfrak{d} \times m}$ according to Algorithm~\ref{CAVI_gmm}, where $$\mathfrak{d} = (B+1+1) \times \left(\underbrace{NK}_{q^{*}(\bZ)} + \underbrace{K}_{q^{*}(\bpi)} + \underbrace{(p+1)K}_{q^{*}(\bmu|\bLam)} + \underbrace{(p^2 + 1)K}_{q^{*}(\bLam)}\right).$$ This accounts for storing posterior parameters from the full data VB, VB on the first split dataset, and $B$ bootstrap replicates. Therefore, we obtain $\mathfrak{d} = (B+2) \times (N+p+p^2+3)K$.

\subsection{Bayesian Mixture Linear Model}
\label{bmlr}

In Bayesian mixture linear model (BMLR), we denote the observation data $(\bY, \mathcal{X})$, where $\bY \in \mathbb{R}^{N \times J_{n}}$ means $N$ response vectors with $\by_{n} \in \mathbb{R}^{J_{n}}$, and $\mathcal{X} \in \mathbb{R}^{N \times J_{n} \times p}$ is a $3$-dimensional tensor with each slice $\bX_{n} \in \mathbb{R}^{J_{n} \times p}$ being a design matrix of $p$ covariates. Therefore, we have totally $N$ datasets $(\by_{n}, \bX_{n})$ for classical linear regression and each of them belongs to one of $K$ clusters, distinguished by the regression coefficient $\bet_{k}$ for $k = 1, \cdots, K$. 

Let $\bZ \in \mathbb{R}^{N \times K}$ be the membership matrix, and $\bpi = (\pi_{1}, \cdots, \pi_{K})^{\prime}$ the mixing coefficient. Then we have prior and likelihood distribution as follows
\begin{equation}\label{prior_bmlr}
\begin{aligned}
    & p\left(\bZ | \bpi\right) = \prod_{n = 1}^{N}\prod_{k = 1}^{K} \pi_{k}^{z_{nk}}, \quad p\left(\bet | \btau\right) = \prod_{k = 1}^{K} N\left(\bet_{k} | 0, \tau_{k}^{-1}\bI_{p}\right)\\
    & p\left(\bpi\right) = \mathrm{Dir}(\bpi | \bal_{0}) = C(\bal_{0})\prod_{k = 1}^{K}\pi_{k}^{\alpha_{0} - 1}, \quad p\left(\btau\right) = \prod_{k = 1}^{K}\mathrm{Gamma}(\tau_{k} | a_{0}, b_{0})\\
    & p\left(\bY | \bZ, \bet, \mathcal{X}\right) = \prod_{n = 1}^{N}\prod_{k = 1}^{K} N\left(\by_{n} | \bX_{n}\bet_{k}, \lambda^{-1}\bI_{J_{n}}\right)^{z_{nk}},
\end{aligned}
\end{equation}
where $\bal_{0} = (\alpha_{0}, \cdots, \alpha_{0})^{\prime}, \btau = (\tau_{1}, \cdots, \tau_{K})^{\prime}, a_{0}, b_{0}, \lambda$ are hyperparameters and $\lambda$ as precision of noise level that is assumed to be known. Similar to the GMM example, the joint distribution under the fractional posterior framework with smoothing parameter $\omega$ is
\begin{equation}\label{joint_bmlr}
    p\left(\bY, \bZ, \bpi, \bet, \btau | \mathcal{X}\right) = p^{\omega}\left(\bY | \bZ, \bet, \mathcal{X}\right)p^{\omega}\left(\bZ | \bpi\right)p\left(\bpi\right)p\left(\bet | \btau\right)p\left(\btau\right).
\end{equation}
When we assume the factorization of mean-field variational approximation to be $q(\bZ, \bet, \btau, \bpi) = q(\bZ)q(\bpi, \bet)q(\btau)$, the corresponding $q^{*}$ that maximizes ELBO is 
\begin{equation}\label{elbo_bmlr}
    q^{*} = \underset{q \in \mathcal{Q}}{\arg\max} \; \int_{\Theta} q(\bZ)q(\bpi, \bet)q(\btau) \log \left(\dfrac{p\left(\bY, \bZ, \bpi, \bet, \btau | \mathcal{X}\right)}{q(\bpi, \bet)q(\btau)q^{\omega}(\bZ)}\right) d\bth,
\end{equation}
for $\bth = (\bZ, \bpi, \bet, \btau)$ in this example. 

Algorithm \ref{CAVI_bmlr} summarizes the procedure of estimating posterior parameters of this model. Similarly, to apply the grid-search TVB method with $N$ samples, $K$ clusters, and $p$-dimensional data on a grid $\boldsymbol{\omega} \in (0, 1]^{m}$ with $B$ bootstrap replicates, we require a dictionary $\mathfrak{D} \in \mathbb{R}^{\mathfrak{d} \times m}$ according to Algorithm~\ref{CAVI_bmlr}, where $$\mathfrak{d} = (B+1+1) \times \left(\underbrace{NK}_{q^{*}(\bZ)} + \underbrace{K}_{q^{*}(\bpi)} + \underbrace{2K}_{q^{*}(\btau)} + \underbrace{(p^2 + p)K}_{q^{*}(\bet)}\right).$$ This accounts for storing posterior parameters from the full data VB, VB on the first split dataset, and $B$ bootstrap replicates. Therefore, we obtain $\mathfrak{d} = (B+2) \times (N+p+p^2+3)K$. The detailed derivation as well as evaluation of ELBO can be found in supplemental materials.

\begin{algorithm}[h!]
\caption{Mean-Field Variational Approximation of Bayesian Mixture Linear Regression Model with Fractional Posterior}\label{CAVI_bmlr}
\begin{algorithmic}[1]

\REQUIRE Observation data $(\bY, \mathcal{X})$, number of clusters $K$, hyperparameters $\bal_{0} = (\alpha_{0}, \cdots, \alpha_{0})^{\prime}$, $a_{0}, b_{0}$, noise level $\lambda$, smoothing parameter $\omega$, maximum number of iterations $\mathrm{max\_iter}$, tolerance $\epsilon > 0$. Initialization: coverage = FALSE, posterior parameters $\left(\alpha_{k}^{(0)}, \bm_{k}^{(0)}, \bS_{k}^{(0)}, a_{k}^{(0)}, b_{k}^{(0)}\right)$ for $k = 1, \cdots, K$, iteration number $t = 0$.

\WHILE{$t < \mathrm{max\_iter}$ \textbf{and} \textbf{not} coverage}
    \STATE Update posterior parameters of $q^{*}(\bZ)$ by
    \begin{equation}\label{update_qz_bmlr}
        r_{nk} = \dfrac{\rho_{nk}}{\sum_{j = 1}^{K}\rho_{nj}},
    \end{equation}
    where $\log \rho_{nk} = \mathbb{E}_{\bpi}\left[\log \pi_{k}\right] - \dfrac{\lambda}{2}\mathbb{E}_{\bet}\left[(\by_{n} - \bX_{n}\bet_{k})^{\prime}(\by_{n} - \bX_{n}\bet_{k})\right]$

    \STATE For $k = 1, \cdots, K$, update posterior parameters of $q^{*}(\tau_{k})$ by
    \begin{equation}\label{update_qtau_bmlr}
        a_{k} = \dfrac{p}{2} + a_{0}, \quad b_{k} = \dfrac{1}{2}\mathbb{E}_{\bet}[\bet_{k}^{\prime}\bet_{k}] + b_{0}
    \end{equation}
    
    \STATE For $k = 1, \cdots, K$, update posterior parameters of $q^{*}(\pi_{k})$ by
    \begin{equation}\label{update_qpi_bmlr}
        \alpha_{k} = \alpha_{0} + \omega \cdot N_{k},
    \end{equation}
    where $N_{k} = \sum\limits_{n = 1}^{N} r_{nk}$

    \STATE For $k = 1, \cdots, K$, update posterior parameters of $q^{*}(\bet_{k})$ by
    \begin{equation}\label{update_qbeta_bmlr}
        \bm_{k} = \omega \cdot \lambda \bS_{k}\sum_{n = 1}^{N}r_{nk}\bX_{n}^{\prime}\by_{n}, \quad \bS_{k} = \left(\mathbb{E}_{\tau_{k}}[\tau_{k}]\bI_{p} + \omega \cdot \lambda \sum_{n = 1}^{N}r_{nk} \bX_{n}^{\prime}\bX_{n}\right)^{-1}
    \end{equation}
\ENDWHILE

\STATE Return $q^{*}(\bZ) = \prod\limits_{n = 1}^{N}\prod\limits_{k = 1}^{K} r_{nk}^{z_{nk}}$, $q^{*}(\bpi) = \mathrm{Dir}(\bpi | \bal)$, $q^{*}(\btau) = \prod\limits_{k = 1}^{K} \mathrm{Gamma}(\tau_{k} | a_{k}, b_{k})$ and $q^{*}(\bet) = \prod\limits_{k = 1}^{K} q^{*}(\bet_{k}) = \prod\limits_{k = 1}^{K} N\left(\bet_{k} | \bm_{k}, \bS_{k}\right)$.
\end{algorithmic}
\end{algorithm}

\section{Empirical Results}
\label{empirical}
In this section, we present numerical results for the two models discussed in Section \ref{example}. We demonstrate that our proposed TVB approach with either sequential update (strategy I) or grid-search update (strategy II) achieves superior calibration of credible intervals compared to original VB estimation. 

In Section \ref{numerical_gmm}, we evaluate our method using both updating strategies and compare it against the standard variational Bayes approach (corresponding to $\omega = 1$ in \eqref{gibbs_post}) for Gaussian mixture model. For the Bayesian mixture linear regression example in Section \ref{numerical_bmlr}, we implement only the grid-search TVB and compare its coverage with standard VB results, as sequential update TVB yields similar coverage results to grid-search TVB. Additionally, for both examples, we compare the TVB calibration with the general posterior calibration (GPC) method  \citep{syring2019calibrating}. As discussed earlier, the original GPC method is not specifically designed for VB, and its calibration approach does not account for the use of fractional VB in the presence of latent variables. We adapt it to incorporate the fractional VB as we found that this adapted version outperforms standard GPC in our setting. Therefore, we report  results for this `fractional GPC', referred to hereafter as GPC. The comparative results between TVB and GPC are presented in the supplementary materials.

For grid-search TVB, we construct a grid of 100 points $\boldsymbol{\omega} = (\omega_{1}, \cdots, \omega_{100})^{\prime}$ by evenly spacing values between $\log 0.001$ and $\log 1$ on the logarithmic scale and then applying the exponential function. For each model and method, we conducted $n = 500$ independent replications to obtain the frequentist coverage probability. 

\subsection{Numerical Study using Gaussian Mixture Model}
\label{numerical_gmm}
For Gaussian Mixture Model described in Section~\ref{gmm}, we use $p = 2$ covariates and $K = 2$ Gaussian clusters $N\left(\bmu_{1}, \Sigma_{1}\right)$ and $N\left(\bmu_{2}, \Sigma_{2}\right)$, where $\bmu_{1} = (0, 0)^{\prime}, \bmu_{2} = (2, 2)^{\prime}, \Sigma_{1} = \Sigma_{2} = \bI_{2}$. We set the mixing parameter $\pi = 0.65$, and vary the sample size $N \in \{1000, 1500, 2000, 2500, 3000\}.$ 

For the traditional VB method, we implement Algorithm \ref{CAVI_gmm} with $\omega = 1$ for the variational approximation of posterior distribution. For calibrated methods (TVB with both updating strategies), we first apply Algorithm \ref{ssb_cal} and Algorithm \ref{ssb_grid} with nominal level $1 - \alpha = 0.95$ to select the proper smoothing parameter $\omega_{0}$, and then use Algorithm \ref{CAVI_gmm} to obtain the variational posteriors and credible intervals.

In this example, we investigate the frequentist coverage performance of estimated mixing parameter $\pi = 0.65$. Since the variational posterior of $\bpi = (\pi, 1 - \pi)^{\prime}$ follows the Dirichlet distribution with posterior parameter $\bal = (\alpha_{1}, \alpha_{2})^{\prime}$, which reduces to that  $\pi$ follows the beta distribution with parameters $(\alpha_{1}, \alpha_{2})$. The corresponding $100(1 - \alpha)\% = 95\%$ credible interval is constructed by $\left[F_{\mathrm{Beta}(\alpha_{1}, \alpha_{2})}^{-1}\left(0.5\alpha\right), F_{\mathrm{Beta}(\alpha_{1}, \alpha_{2})}^{-1}\left(1 - 0.5\alpha\right)\right]$ where $F$ is the corresponding cumulative distribution function. Note that in view of the label switching issue, we always let $\hat{\pi} = \max(\hat{\pi}_{1}, \hat{\pi}_{2})$, where $\hat{\bpi} = (\hat{\pi}_{1}, \hat{\pi}_{2})^{\prime}$ is the posterior mean.

\begin{figure}[ht!]
    \centering
    \includegraphics[width = 14cm]{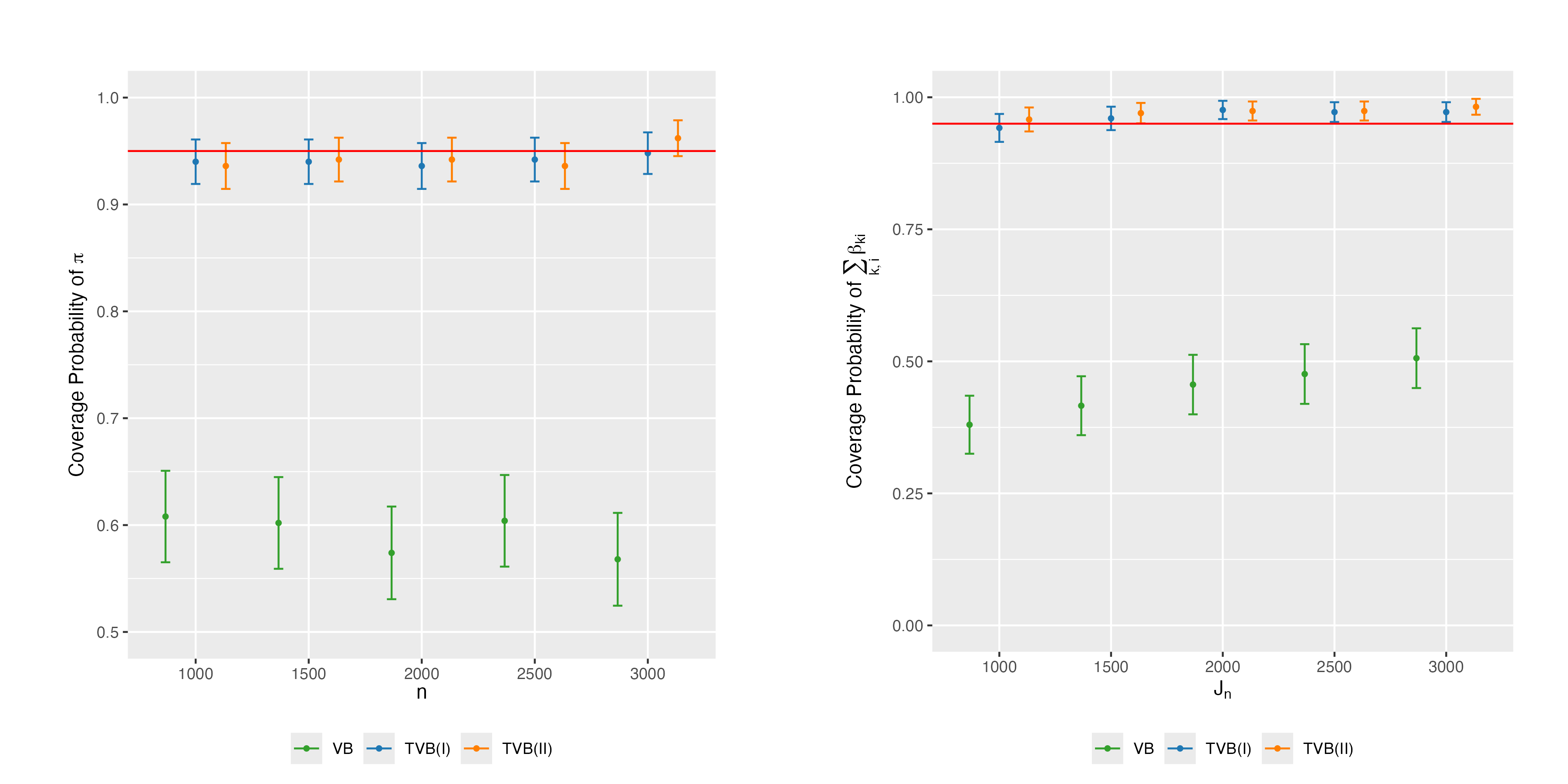}
    \caption{Frequentist coverage performance for original VB method (VB) and trustworthy VB method (TVB). The left panel displays the coverage of the mixing parameter $\pi$ in the GMM example, comparing TVB with both updating strategies. The right panel presents the coverage of the sum of regression coefficients, $\sum\limits_{k=1}^{K}\sum\limits_{i=1}^{p}\beta_{ki}$ in the BMLR example, using TVB with grid-search update only. In each panel, dot plots with error bars indicate the average coverage and the 95\% confidence interval.}
    \label{fig:plot_cover}
\end{figure}

\begin{figure}[ht!]
    \centering
    \includegraphics[width = 14cm]{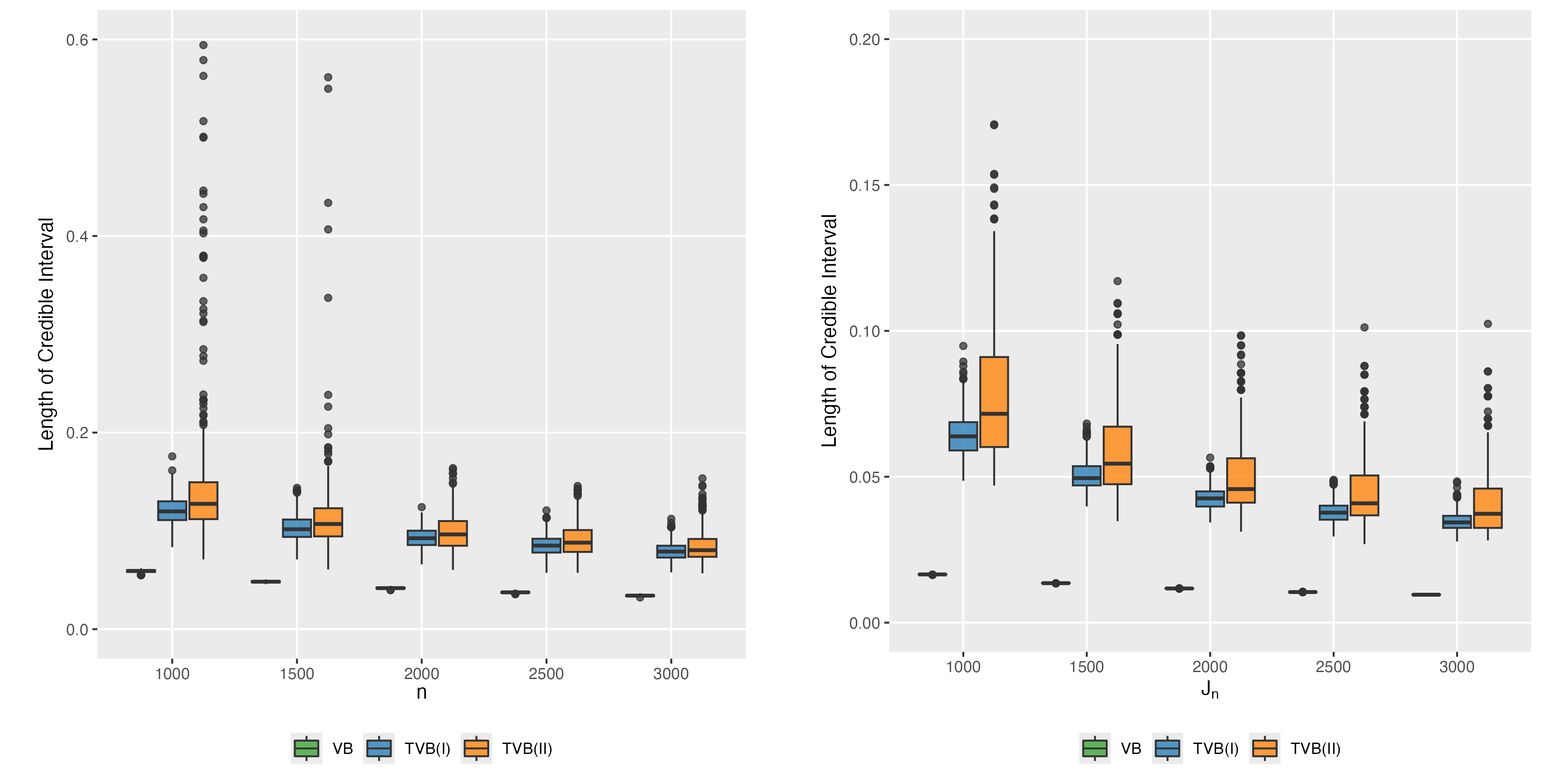}
    \caption{Boxplots of credible interval length for original VB method (VB) and trustworthy VB method (TVB). The left panel is for $\pi$ in the GMM example, comparing TVB with both updating strategies; the right panel is for $\sum\limits_{k = 1}^{K}\sum\limits_{i = 1}^{p}\beta_{ki}$ in the BMLR example, using TVB with grid-search update only.}
    \label{fig:plot_length}
\end{figure}

\begin{figure}[ht!]
    \centering
    \includegraphics[width = 14cm]{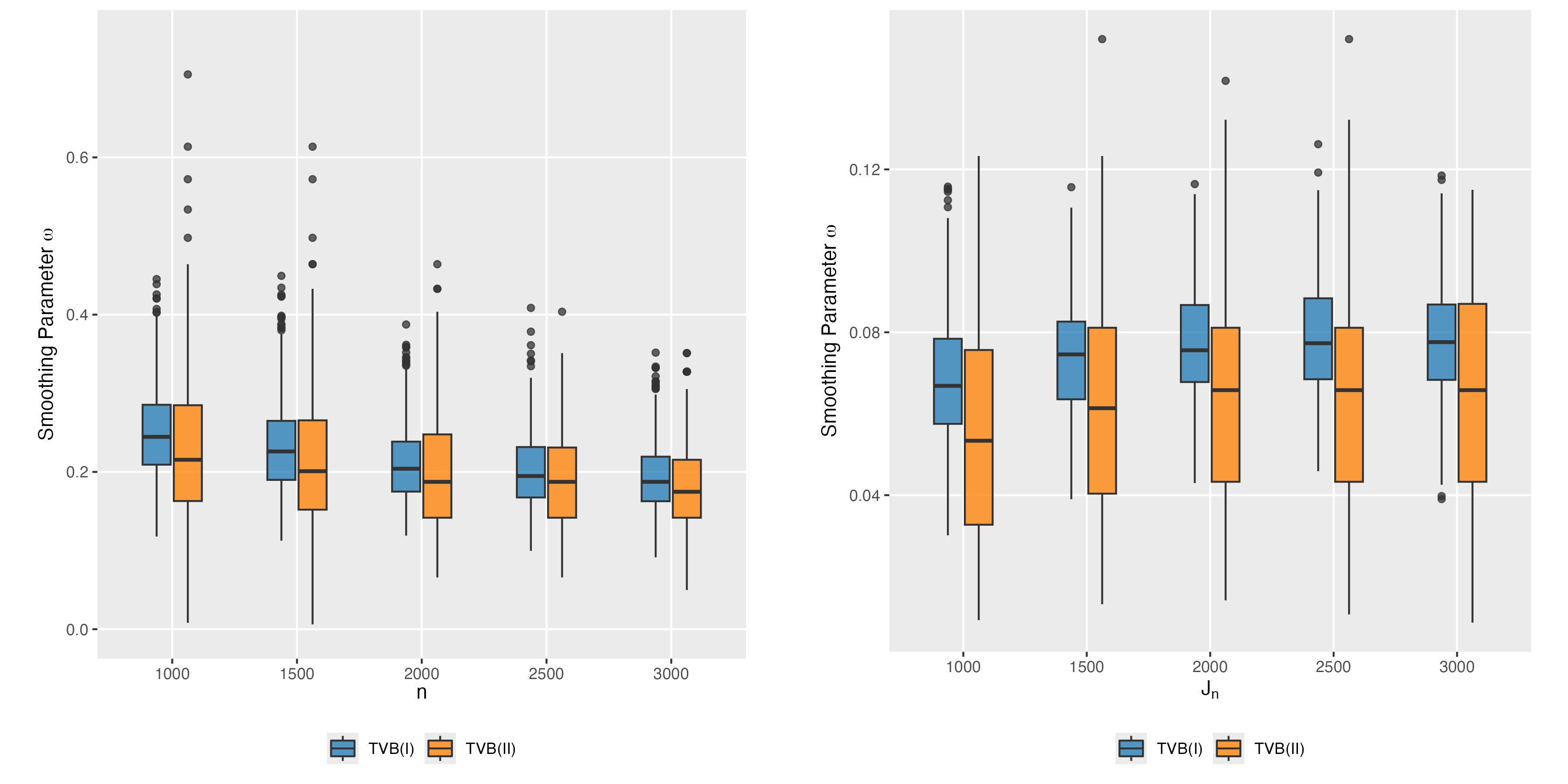}
    \caption{Boxplots of the selected smoothing parameter $\omega$ for trustworthy VB method (TVB). The left panel is for the GMM example comparing TVB with both updating strategies, while the right panel is for the BMLR example using TVB with grid-search update only.}
    \label{fig:plot_omega}
\end{figure}

As shown in the left panel of Figure \ref{fig:plot_cover}, the TVB with both updating strategies successfully calibrates credible intervals to the nominal $95\%$ level across varying sample sizes. This calibration performance demonstrates similar accuracy and robustness of both updating strategies, while the standard VB approach consistently exhibits substantial undercoverage. As expected, credible interval of the TVB method is larger than that of the original VB ones; see Figure \ref{fig:plot_length} (left panel) and \ref{fig:plot_omega} (left panel) for the boxplot of selected $\omega_{0}$, respectively. The VB running time for this simulation when $N = 3000$ is 4.52 seconds (averaged over 20 repeats), demonstrating that our TVB method is scalable given sufficient computational resources.

\subsection{Numerical Study of Bayesian Mixture Linear Model}
\label{numerical_bmlr}
Consider the Bayesian mixture linear model discussed in Section \ref{bmlr}, we set $N = 500$ datasets with the observation of each dataset $(\by_{i}, \bX_{i})$ to be $J_{n} = (1000, 1500, 2000, 2500, 3000)$. We again let $K = 2$ as the number of mixed clusters with mixing probability $\pi = 0.65$, and for each cluster, the regression coefficient $\bet_{k} \in \mathbb{R}^{p}$ for $k = 1, 2$ is generated from $N\left(\begin{pmatrix} 0\\0 \end{pmatrix}, \begin{pmatrix}
    \tau_{k}^{-1} & 0\\ 0 & \tau_{k}^{-1}
\end{pmatrix}\right)$, where $p = 2$ and $\tau_{1} = \tau_{2} = 1$. Also, we generate noise $\boldsymbol{\epsilon}_{k} \in \mathbb{R}^{J_{n}}, k = 1, 2$ as the noise term of each cluster by $\epsilon_{kj} \overset{i.i.d.}{\sim} N(0, 1), j = 1, \cdots, J_{n}$. For each $i = 1, \cdots, N$, we first obtain the membership matrix $\bZ$ randomly based on $\pi$, suppose $i$th dataset $(\by_{i}, \bX_{i})$ belongs to $k$th cluster for $k = 1, 2$. To get such dataset, we let $\bX_{i} \in \mathbb{R}^{J_{n} \times p}$ be from iid $N(0, 1)$, and then $\by_{i} = \bX_{i}\bet_{k} + \tilde{\boldsymbol{\epsilon}}_{k}$, where $\tilde{\boldsymbol{\epsilon}}_{k} = \dfrac{\|\bX_{i}\bet_{k}\|_{2}^{2}}{\mathrm{SNR}} \cdot \dfrac{\boldsymbol{\epsilon}_{k}}{\|\boldsymbol{\epsilon}_{k}\|_{2}^{2}}$ for a prespecified signal-to-noise ratio (SNR). Throughout our simulation, we let the signal-to-noise ratio to be $0.1$.

For this example, we want to build $100(1 - \alpha)\% = 95\%$ credible interval for $\sum\limits_{k = 1}^{K}\sum\limits_{i = 1}^{p}\beta_{ki}$. According to Algorithm \ref{CAVI_bmlr}, the posterior distribution of $\bet_{k}$ is Gaussian with mean $\bm_{k}$ and variance $\bS_{k}$, and therefore we know the posterior distribution of $\sum\limits_{k = 1}^{K}\sum\limits_{i = 1}^{p}\beta_{ki}$ should be $N(\sum\bm, \sum\bS)$, where $\sum\bm, \sum\bS$ represents the summation of all elements of $\bm_{k}, \bS_{k}$ across $k$, respectively. Therefore, the $100(1 - \alpha)\%$ credible interval of $\sum\limits_{k = 1}^{K}\sum\limits_{i = 1}^{p}\beta_{ki}$ is constructed by $\left[\sum\bm - z_{1-0.5\alpha}\sqrt{\sum\bS}, \sum\bm + z_{1-0.5\alpha}\sqrt{\sum\bS}\right]$ where $z_{1-0.5\alpha}$ is $1 - 0.5\alpha$ quantile of $N(0, 1)$.

It is shown in the right panel of Figure \ref{fig:plot_cover} that again, TVB method with grid-search update successfully calibrates credible intervals to the nominal level, and has better and more robust performance than traditional VB. The length of credible interval for TVB method is larger than original VB ones, as shown in Figure \ref{fig:plot_length} (right panel). The boxplot of selected $\omega_{0}$ values is shown in \ref{fig:plot_omega} (right panel). This finding aligns with our previous results from the Gaussian mixture model example. The VB running time in this simulation when $N = 500, J_{n} = 3000$ is 0.491 seconds (averaged over 20 repeats), demonstrating that our TVB method is scalable given sufficient computational resources.

\section{Real Data Analysis}
\label{real_data}

In this section, we apply grid-search TVB method based on Gaussian Mixture Models (GMM) to analyze the \texttt{faithful} dataset from R's \texttt{datasets} package \cite{R_datasets}. The \texttt{faithful} dataset contains 272 observations of the Old Faithful geyser in Yellowstone National Park, with two variables: \texttt{eruptions} (the duration of eruptions in minutes) and \texttt{waiting} (the waiting time until the next eruption in minutes) \citep{azzalini1990look}. Our analytical approach involved fitting a two-component GMM to identify the underlying bimodal distribution structure of the geyser's behavior, facilitating the classification of eruptions into short/long categories and corresponding waiting times. The analysis revealed distinct clusters in the eruption patterns, with longer eruptions typically followed by longer waiting periods, suggesting a potential physical mechanism related to the geyser's underground reservoir refilling process. As shown in Figure \ref{fig:realdata} left panel, the visualization of these clusters, obtained by VB with $K = 2$, clearly demonstrates the bimodal nature of Old Faithful's eruption cycles, with the GMM effectively capturing the natural groupings in the data.

\begin{figure}[ht!]
    \centering
    \includegraphics[width = 14cm]{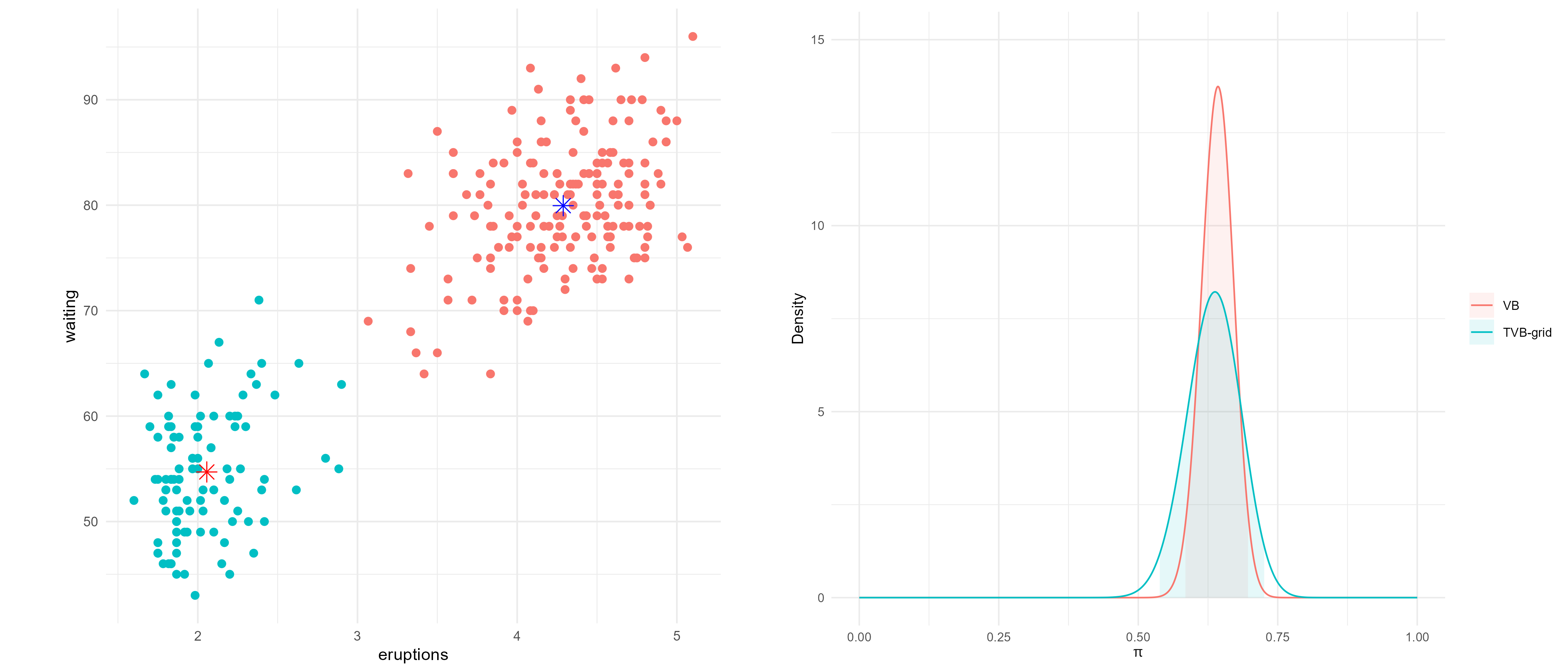}
    \caption{TVB recalibration of the mixing parameter for the \texttt{faithful} dataset. The left panel is \texttt{faithful} dataset clustered by variational Bayes with $K = 2$; The right plot is the posterior distribution of VB and grid-search TVB of mixing parameter $\pi$ with 95\% credible intervals.}
    \label{fig:realdata}
\end{figure}

First, we implement original variational Bayes and grid-search TVB with $K = 2$ for mixing probability $\pi$ (for the larger cluster), 
where we use a grid of 500 points $\boldsymbol{\omega} = (\omega_{1}, \cdots, \omega_{500})^{\prime}$ by evenly spacing values between $\log 0.001$ and $\log 1$ on the logarithmic scale and then applying the exponential function. We build the 95\% credible interval for mixing probability under each method: The VB 95\% CI is $(0.584, 0.698)$ while the grid-search TVB 95\% CIs is $(0.538, 0.726)$. As we can see, after TVB calibration, the 95\% credible interval is wider which increases the trustworthiness of the corresponding uncertainty quantification. Furthermore, the grid-search TVB approach demonstrates significant computational efficiency, as it requires only a single bootstrap estimation procedure, with results stored as a comprehensive "dictionary" for subsequent uncertainty quantification of any target parameter. In contrast, the sequential update TVB necessitates repeated bootstrap simulations whenever inference on different parameters is required, resulting in substantially higher computational overhead.

To better illustrate the computational efficiency of grid-search TVB, we next obtain calibrated 95\% credible intervals for $\bmu_{ij}, i, j = 1, 2$ and $\sum\limits_{j = 1}^{2}\bmu_{ij}, i = 1, 2$, where $\bmu_{ij}$ is the $j$-th coordinate of the mean vector of the $i$-th cluster. Recall that in Algorithm \ref{CAVI_gmm}, the posterior distribution of $\bmu_{i}$ conditional on $\bLam_{i}$ follows a Gaussian-Wishart distribution. Therefore, following \cite{murphy2007conjugate}, the marginal distribution of $\bmu_{i}$ is a multivariate $t$-distribution with degrees of freedom $\nu_{i} - p + 1$, location vector $\bm_{i}$, and scale matrix $\frac{\bW_{i}^{-1}}{\beta_{i}(\nu_{i} - p + 1)}$. As shown in Table \ref{table:real_data}, for each parameter of interest, we obtain calibrated 95\% credible intervals that are wider than or equal to the original VB intervals, demonstrating substantial improvements in coverage properties for all target parameters relative to the original VB approach. Additionally, by using the TVB dictionary constructed in the grid-search TVB method, we avoid having to re-run the entire update procedure (as required by the sequential TVB method) for different target parameters, making grid-search TVB more efficient and therefore recommended.

\begin{table}[ht!]
\centering
\begin{tabular}{c c c c}
\hline
& VB & TVB & $\omega$ \\
\hline
$\bmu_{11}$ & (4.22, 4.35) & (4.22, 4.35) & 1.000 \\
\hline
$\bmu_{12}$ & (79.06, 80.85) & (78.90, 81.01) & 0.717 \\
\hline
$\bmu_{21}$ & (1.98, 2.13) & (1.90, 2.35) & 0.265 \\
\hline
$\bmu_{22}$ & (53.49, 55.91) & (52.01, 60.37) & 0.127 \\
\hline
$\sum\limits_{j = 1}^{2}\bmu_{1j}$ & (83.33, 85.16) & (83.17, 85.31) & 0.737 \\
\hline
$\sum\limits_{j = 1}^{2}\bmu_{2j}$ & (55.52, 58.00) & (55.05, 61.81) & 0.124 \\
\hline
\end{tabular}
\caption{95\% credible intervals for variational Bayes (VB) and grid-search TVB for different target parameters.}
\label{table:real_data}
\end{table}

\section{Discussion}
\label{discussion}

This paper introduces a novel approach to recalibrate credible intervals in variational Bayes inference through fractional variational Bayes and sample splitting bootstrap. Our method successfully addresses the undercoverage issue common in VB posterior approximations while maintaining computational efficiency. The proposed method demonstrates superior empirical performance compared to standard VB, and we provide theoretical guarantees for the coverage probability of calibrated VB credible intervals. Overall, our work enables trustworthy VB with strong empirical and theoretical guarantees. 

There are a few future research directions building on our work. First, 
while our method achieves guaranteed calibration in moderate to large samples, there is increased variance in small-sample settings due to the sample splitting procedure. However, this limitation is well-understood given the inherent variance-bias trade-off: sample splitting method will increase variance when we try to reduce the bias. Future work could explore adaptive splitting ratios or alternative resampling schemes that are particularly well suited for small samples. 
 
In addition, our method can be extended to recalibrate other approximating methods with misspecified models, as well as complex data structures with temporal or spatial dependencies.

\section{Proof}
\label{proof}
In this section, we present the proofs of Theorem~\ref{converge}, Corollary~\ref{cor:uniform}, and Theorem~\ref{cor_coverage}. We begin with a simple technical result:
\begin{lem}
\label{lemma_smooth}
    For random variable $X \sim f_{X}(x)$, suppose $\sigma(z) = \dfrac{1}{1 + e^{-kz}}$ is sigmoid function. Then for any $\theta \in \mathbb{R}$, we have
    \begin{equation}
        \mathbb{E}_{f_X}\left[\sigma(X - \theta)\right] \rightarrow \mathbb{P}\left(\theta < X\right),
    \end{equation}
    as $k \rightarrow \infty$. Therefore, $\sigma(X - \theta)$ is a smooth surrogate of the indicator function $\mathbbm{1}\{\theta < X\}$.
\end{lem}
\begin{proof}[Proof of Lemma~\ref{lemma_smooth}]
    For $\sigma(X - \theta) = \dfrac{1}{1 + e^{-k(X - \theta)}}$, as $k \rightarrow \infty$, if $X > \theta$, then $\sigma(X - \theta) \rightarrow 1$, and if $X < \theta$, then $\sigma(X - \theta) \rightarrow 0$. Therefore, we have
    \begin{equation}
        \mathbb{E}_{f_X}\left[\sigma(X - \theta)\right] = \int_{-\infty}^{\infty}\dfrac{1}{1 + e^{-k(x - \theta)}}f_{X}(x)dx \rightarrow \int_{\theta}^{\infty}f_{X}(x)dx = \mathbb{P}\left(\theta < X\right),
    \end{equation}
    as $k \rightarrow \infty$. This completes of the proof. 
\end{proof}

\begin{proof}[Proof of Theorem \ref{converge}] For $X_{1}, \cdots, X_{n} \overset{i.i.d.}{\sim} P_{\bth}$, denote the empirical measure by $\mathbb{P}_{n} = \dfrac{1}{n}\sum\limits_{i = 1}^{n}\delta_{X_{i}}$, where $\delta_{X_{i}}$ is the Dirac delta measure at $X_{i}$. After splitting the sample into $\bX_{1}$ and $\bX_{2}$, let $X_{2, 1}^{(b)}, \cdots, X_{2, l}^{(b)}$ be bootstrap samples from $\bX_{2}$ with $l = n/2$ (without loss of generality, let $n$ be an even integer), and denote the corresponding empirical measure by $\mathbb{P}_{l}^{(b)}$. The associated bootstrap empirical process is
\begin{equation}
    \hat{\mathbb{G}}_{n, l} := \sqrt{l}\left(\mathbb{P}_{l}^{(b)} - \mathbb{P}_{n}\right) = \dfrac{1}{\sqrt{l}}\sum_{i = 1}^{n}\left(N_{li} - \dfrac{l}{n}\right)
\end{equation}
as in \cite{wellner2013weak}, where $N_{li}$ is the random variable counting the times $X_{i}$ is resampled. For any measurable $g$,
$\hat{\mathbb{G}}_{n,l}g = l^{-1/2}\sum_{i = 1}^{n}(N_{li} - l/n)g(X_i)$.

Recall that $T_{\bth}(\mathbb{P}_{n})$ and $T_{\bth}(\mathbb{P}_{l}^{(b)})$ are operators on the Donsker class $\mathcal{F} := \left\{f_{\omega}(\cdot): \omega \in (0, 1]\right\}$. Let $l^{\infty}(\mathcal{F})$ be the space of bounded real-valued functions on $\mathcal{F}$, equipped with the norm
$\|T\|_{\mathcal{F}} := \sup_{f_{\omega}\in\mathcal{F}} |T(f_{\omega})|$, for $T \in l^{\infty}(\mathcal{F})$. 
For either $\mathbb{P}_{n}$ or $\mathbb{P}^{(b)}_{l}$ and any $f_{\omega} \in \mathcal{F}$, we have $T_{\bth}(\mathbb{P}_{n})f_{\omega}, T_{\bth}(\mathbb{P}^{(b)}_{l})f_{\omega} \in (0, 1]$. Therefore, with fixed $\bth$, $T_{\bth}(\mathbb{P}_{n}), T_{\bth}(\mathbb{P}^{(b)}_{l}) \in l^{\infty}(\mathcal{F})$.

Define $\psi:l^{\infty}(\mathcal{F})\to\mathbb{R}$ by
\[
\psi(T) := \sup_{f_{\omega}\in\mathcal{F}} |T(f_{\omega})|.
\]
Then $\psi(T)\in(0,1]$ for the operators of interest, and for any
$T,\tilde T\in l^{\infty}(\mathcal{F})$,
\[
\bigl|\psi(T) - \psi(\tilde T)\bigr|
=
\Bigl|\sup_{f_{\omega}\in\mathcal{F}}|T(f_{\omega})|
      - \sup_{f_{\omega}\in\mathcal{F}}|\tilde T(f_{\omega})|\Bigr|
\le
\sup_{f_{\omega}\in\mathcal{F}}|T(f_{\omega}) - \tilde T(f_{\omega})|
= \|T - \tilde T\|_{\mathcal{F}}.
\]
Thus $\psi$ is a bounded Lipschitz functional on $l^{\infty}(\mathcal{F})$.

By Assumption \ref{Donsker} and Assumption \ref{hadamard}, according to Theorem 3.7.3 and Theorem 3.10.11 in \cite{wellner2013weak}, we have
\begin{equation}
    \left|\mathbb{E}\psi\left(\sqrt{l}\left(T_{\bth}(\mathbb{P}_{l}^{(b)}) - T_{\bth}(\mathbb{P}_{n})\right)\right) - \mathbb{E}\psi\left(T_{\bth, P_{\bth}}^{\prime}(\mathbb{G})\right)\right| \overset{P}{\rightarrow} 0,
\end{equation}
where $\mathbb{G}$ is a tight Brownian bridge with mean 0 and $T'_{\bth,P_{\bth}}$ is the Hadamard derivative in
Assumption~\ref{hadamard}.

Since $T'_{\bth,P_{\bth}}(\mathbb{G})$ is a tight mean-zero Gaussian process
in $l^{\infty}(\mathcal{F})$, the random variable
$\psi\bigl(T'_{\bth,P_{\bth}}(\mathbb{G})\bigr)
= \sup_{f_{\omega}\in\mathcal{F}}|T'_{\bth,P_{\bth}}(\mathbb{G})f_{\omega}|$
has finite expectation, so
$\mathbb{E}\,\psi\bigl(T'_{\bth,P_{\bth}}(\mathbb{G})\bigr) < \infty$ and hence
$\mathbb{E}\,\psi\bigl(\sqrt{l}\{T_{\bth}(\mathbb{P}_{l}^{(b)}) - T_{\bth}(\mathbb{P}_{n})\}\bigr)
= O_{P}(1)$, that is,
\[
\mathbb{E}
\sup_{f_{\omega}\in\mathcal{F}}
\Bigl|
\sqrt{l}\bigl\{T_{\bth}(\mathbb{P}_{l}^{(b)}) f_{\omega}
            - T_{\bth}(\mathbb{P}_{n}) f_{\omega}\bigr\}
\Bigr|
= O_{P}(1). 
\]
Since $l = n/2 \rightarrow \infty$,

the preceding display yields 
\begin{equation}
    \left|\mathbb{E}\underset{f_{\omega} \in \mathcal{F}}{\sup}\left|\left(T_{\bth}(\mathbb{P}_{l}^{(b)})f_{\omega} - T_{\bth}(\mathbb{P}_{n})f_{\omega}\right)\right|\right| \overset{P}{\rightarrow} 0.
\end{equation}
This, together with the triangle inequality, implies that $\left|\left(\mathbb{E}T_{\bth}(\mathbb{P}_{l}^{(b)})f_{\omega} - \mathbb{E}T_{\bth}(\mathbb{P}_{n})f_{\omega}\right)\right| \overset{P}{\rightarrow} 0$ for any $f_{\omega} \in \mathcal{F}$. Invoking Lemma \ref{lemma_smooth} leads to that $\mathbb{E}T_{\bth}(\mathbb{P}_{l}^{(b)})f_{\omega} \rightarrow \mathbb{P}\left(h(\bth) \in \left(L\left(\mathbb{P}_{l}^{(b)}\right)f_{\omega}, U\left(\mathbb{P}_{l}^{(b)}\right)f_{\omega}\right)\right)$ when $k \rightarrow \infty$. Therefore we have shown that 
\begin{equation}
    \left|\mathbb{P}\left(h(\bth) \in \left(L\left(\mathbb{P}_{l}^{(b)}\right)f_{\omega}, U\left(\mathbb{P}_{l}^{(b)}\right)f_{\omega}\right)\right) - \mathbb{P}\left(h(\bth) \in \left(L\left(\mathbb{P}_{n})f_{\omega}, U(\mathbb{P}_{n}\right)f_{\omega}\right)\right)\right| \overset{P}{\rightarrow} 0,
\end{equation}
for $\forall f_{\omega} \in \mathcal{F}$ and as $n \rightarrow \infty$.

For any fixed empirical measure $\mathbb{P}_{n}$, since both the sigmoid function $\sigma(\cdot)$ and $h(\cdot)$ are continuous functions, $T_{\bth}(\mathbb{P}_{n})f_{\omega}$ is also a continuous function of $\bth$. Let the variational posterior mean obtained using $\bX_{1}$ be $\hat{\bth}^{\bX_{1}}_{\omega}$. According to Assumption \ref{assump_consistency}, there holds $h(\hat{\bth}^{\bX_{1}}_{\omega}) \rightarrow h(\bth)$ in probability as $n \rightarrow \infty$. Then it follows that $\left|T_{\bth}(\mathbb{P}_{l}^{(b)})f_{\omega} - T_{\hat{\bth}^{\bX_{1}}_{\omega}}(\mathbb{P}_{l}^{(b)})f_{\omega}\right| = o_{p}(1)$, and consequently as $n \rightarrow \infty$, for any $f_{\omega} \in \mathcal{F}$, we have 
\begin{equation}\left|\mathbb{E}T_{\hat{\bth}^{\bX_{1}}_{\omega}}(\mathbb{P}_{l}^{(b)})f_{\omega} - \mathbb{E}T_{\bth}(\mathbb{P}_{n})f_{\omega}\right| \overset{P}{\rightarrow} 0, 
\end{equation}
i.e., 
\begin{equation}\left|\mathbb{P}\left(h(\hat{\bth}^{\bX_{1}}_{\omega}) \in \left(L\left(\mathbb{P}_{l}^{(b)}\right)f_{\omega}, U\left(\mathbb{P}_{l}^{(b)}\right)f_{\omega}\right)\right) - \mathbb{P}\left(h(\bth) \in \left(L\left(\mathbb{P}_{n})f_{\omega}, U(\mathbb{P}_{n}\right)f_{\omega}\right)\right)\right| \overset{P}{\rightarrow} 0. 
\end{equation}
By the definitions of $c_{q^{*}}^{B}(\omega; \alpha)$ and $c_{q^{*}}(\omega; \alpha)$, this yields
\begin{equation}
    c_{q^{*}}^{B}(\omega; \alpha) \overset{P}{\rightarrow} c_{q^{*}}(\omega; \alpha), \quad n \rightarrow \infty.
\end{equation}

Finally, since $\hat{c}_{q^{*}}^{B}(\omega; \alpha)$ defined in \eqref{eq:c^B_q} is the sample version of $c_{q^{*}}^{B}(\omega; \alpha)$ with $\left|\hat{c}_{q^{*}}^{B}(\omega; \alpha) - c_{q^{*}}^{B}(\omega; \alpha)\right| = O_{p}(B^{-1/2})$, by weak law of large number, we have
\begin{equation}
    \hat{c}_{q^{*}}^{B}(\omega; \alpha) \overset{P}{\rightarrow} c_{q^{*}}(\omega; \alpha), \quad n, B \rightarrow \infty.
\end{equation} 
This completes the proof of Theorem \ref{converge}.
\end{proof}

\begin{proof}[Proof of Corollary~\ref{cor:uniform}] The proof of Corollary~\ref{cor:uniform} follows a similar argument to that of Theorem~\ref{converge}. From the proof of Theorem~\ref{converge}, we have $\left|\mathbb{E}\underset{f_{\omega} \in \mathcal{F}}{\sup}\left|\sqrt{l}\left(T_{\bth}(\mathbb{P}_{l}^{(b)})f_{\omega} - T_{\bth}(\mathbb{P}_{n})f_{\omega}\right)\right|\right| \overset{P}{\rightarrow} 0$. With Assumption~\ref{assump:uniform}, it is straightforward that $h(\hat{\bth}^{\bX_{1}}_{\omega}) \rightarrow h(\bth)$ in probability as $n \rightarrow \infty$ uniformly over $\omega \in \Omega$, and thus 
\begin{equation*}\underset{\omega \in \Omega}{\sup}\left|\mathbb{P}\left(h(\hat{\bth}^{\bX_{1}}_{\omega}) \in \left(L\left(\mathbb{P}_{l}^{(b)}\right)f_{\omega}, U\left(\mathbb{P}_{l}^{(b)}\right)f_{\omega}\right)\right) - \mathbb{P}\left(h(\bth) \in \left(L\left(\mathbb{P}_{n})f_{\omega}, U(\mathbb{P}_{n}\right)f_{\omega}\right)\right)\right| \overset{P}{\rightarrow} 0. 
\end{equation*}
By the definitions of $c_{q^{*}}^{B}(\omega; \alpha)$ and $c_{q^{*}}(\omega; \alpha)$, this yields
\begin{equation}
    \underset{\omega \in \Omega}{\sup}\left|c_{q^{*}}^{B}(\omega; \alpha) - c_{q^{*}}(\omega; \alpha) \right| \overset{P}{\rightarrow} 0, \quad n \rightarrow \infty.
\end{equation}
Since $\hat{c}_{q^{*}}^{B}(\omega; \alpha)$ defined in \eqref{eq:c^B_q} is the sample version of $c_{q^{*}}^{B}(\omega; \alpha)$ with $\left|\hat{c}_{q^{*}}^{B}(\omega; \alpha) - c_{q^{*}}^{B}(\omega; \alpha)\right| = O_{p}(B^{-1/2})$, by weak law of large number, we have
\begin{equation}
    \underset{\omega \in \Omega}{\sup}\left|\hat{c}_{q^{*}}^{B}(\omega; \alpha) - c_{q^{*}}(\omega; \alpha)\right| \overset{P}{\rightarrow} 0, \quad n, B \rightarrow \infty.
\end{equation} 
This completes the proof of Corollary~\ref{cor:uniform}.
\end{proof}

\begin{proof}[Proof of Theorem~\ref{cor_coverage}]
    In proving Theorem~\ref{cor_coverage}, we denote the data-dependent smoothing parameter as $\hat{\omega}_{n}$ (instead of $\hat{\omega}_{0}$) to emphasize its dependence on the sample. Recall that we obtain $\hat{\omega}_{n}$ as an $M$-estimator by solving \eqref{eq:optimize}. From Corollary~\ref{cor:uniform}, we know that $\hat{c}_{q^{*}}^{B}(\omega; \alpha) \rightarrow c_{q^{*}}(\omega; \alpha)$ as $n, B \rightarrow \infty$ uniformly over $\Omega$. By Assumption~\ref{assump:uniqueness}, if we further define $f(\omega) \coloneqq \lim\limits_{n \rightarrow \infty} \mathbb{P}\left(h(\bth) \in C_{q^{*}_{\omega}}(\bX; \alpha) \right)$, then $\hat{c}_{q^{*}}^{B}(\omega; \alpha)$ converges to $f(\omega)$ uniformly over $\Omega$ as $B \rightarrow \infty$. Therefore, by Theorem 5.7 in \cite{van2000asymptotic}, the $M$-estimator $\hat{\omega}_{n}$ converges to $\omega^{*}$ as $n, B \rightarrow \infty$.
    
    Now we show that using the consistent estimator $\hat{\omega}_n$ yields the correct coverage. We want to evaluate the limit of the true coverage probability, $\mathbb{P}\left(h(\bth) \in C_{q^{*}_{\hat{\omega}_{n}}}(\bX; \alpha)\right)$. Consider the difference between this probability and our target value, $1-\alpha = f(\omega^*)$. Using the triangle inequality, we have:
    $$\left|\mathbb{P}\left(h(\bth) \in C_{q^{*}_{\hat{\omega}_{n}}}(\bX; \alpha)\right) - f(\omega^{*})\right| \leq \left|\mathbb{P}\left(h(\bth) \in C_{q^{*}_{\hat{\omega}_{n}}}(\bX; \alpha)\right) - f(\hat{\omega}_n)\right| + \left|f(\hat{\omega}_n) - f(\omega^{*})\right|. $$ Since both terms converge to 0 as $n, B \rightarrow \infty$, we complete the proof.
\end{proof}

\section*{Supplementary Material}
\addcontentsline{toc}{section}{Supplementary Materials}

The supplementary material contains additional technical details and results. Section A provides detailed derivation of Algorithm~\ref{CAVI_gmm} and Algorithm~\ref{CAVI_bmlr}. Section B presents additional simulation studies exploring the comparison between TVB and GPC methods. Code implementing the proposed method, along with illustrative numerical examples and scripts to reproduce selected experiments, is publicly available at
\url{https://github.com/JmL130169/Trustworthy-Variational-Bayes}.

\bibliographystyle{apalike}
\bibliography{reference.bib}

\newpage
\begin{center}
{\large\bfseries Supplementary Material of ``Bend to Mend: Toward Trustworthy Variational Bayes with Valid Uncertainty Quantification''}
\end{center}

This supplementary material provides additional technical details to support the main text. Section~\ref{sec:supp1} contains the derivation of Algorithm~\ref{CAVI_gmm} and Algorithm~\ref{CAVI_bmlr}, and Section~\ref{sec:supp2} presents comparative study between TVB and GPC methods.

\appendix

\section{Derivation of Algorithms}
\label{sec:supp1}
In Appendix \ref{sec:supp1}, we present the complete derivation of Algorithms \ref{CAVI_gmm} and \ref{CAVI_bmlr}, which extend the original variational Bayes approach to fractional variational Bayes. Throughout this section, we use $C$ to denote a generic constant, whose value may change from line to line.

\subsection{Derivation of Algorithm \ref{CAVI_gmm}}
Recall that the general updating rule for mean-field variational Bayes is presented in Equation \eqref{CAVI}. Consequently, based on the model specification in \eqref{prior_gmm}, we derive the following updates.

\begin{itemize}[leftmargin=*]
    \item \textbf{Update of $q(\bZ)$: }

    Utilizing Equations \eqref{CAVI}, \eqref{joint_gmm} and \eqref{elbo_gmm}, we obtain
    \begin{equation}\label{q_z}
    \begin{aligned}
        \omega \log q^{*}\left(\bZ\right) &= \mathbb{E}_{\bpi, \bmu, \bLam}\left\{\log p\left(\bX, \bZ, \bpi, \bmu, \bLam\right)\right\} + C \\
        &= \mathbb{E}_{\bpi, \bmu, \bLam}\left\{\log p^{\omega}\left(\bX | \bZ, \bmu, \bLam\right) + \log p^{\omega}\left(\bZ | \bpi\right) + \log p\left(\bpi\right) + \log p\left(\bmu | \bLam\right) + \log p\left(\bLam\right)\right\} + C \\
        &= \mathbb{E}_{\bmu, \bLam}\left\{\log p^{\omega}\left(\bX | \bZ, \bmu, \bLam\right)\right\} + \mathbb{E}_{\bpi}\left\{\log p^{\omega}\left(\bZ | \bpi\right)\right\} + C \\
        &= \mathbb{E}_{\bmu, \bLam}\left\{\omega \sum_{n = 1}^{N}\sum_{k = 1}^{K} z_{nk} \log N\left(\bx_{n} | \bmu_{k}, \bLam_{k}^{-1} \right)\right\} + \mathbb{E}_{\bpi}\left\{\omega\sum_{n = 1}^{N}\sum_{k = 1}^{K} z_{nk} \log \pi_{k} \right\} + C \\
        &= \omega\sum_{n = 1}^{N}\sum_{k = 1}^{K} z_{nk} \underbrace{\left\{\mathbb{E}_{\bpi}\left[\log \pi_{k}\right] + \dfrac{1}{2}\mathbb{E}_{\bLam}\left[\log |\bLam_{k}|\right] - \dfrac{1}{2}\mathbb{E}_{\bmu, \bLam}\left[(\bx_{n} - \bmu_{k})^{\prime}\bLam_{k}(\bx_{n} - \bmu_{k})\right]\right\}}_{\log \rho_{nk}} + C.
    \end{aligned}
    \end{equation}
    Thus, it follows that $q^{*}(\bZ) \propto \prod\limits_{n = 1}^{N}\prod\limits_{k = 1}^{K} \rho_{nk}^{z_{nk}}$, and by normalizing $\rho_{nk}$ as $r_{nk} = \dfrac{\rho_{nk}}{\sum\limits_{j = 1}^{K} \rho_{nj}}$, we have $q^{*}(\bZ) = \prod\limits_{n = 1}^{N}\prod\limits_{k = 1}^{K} r_{nk}^{z_{nk}}$. From equation \eqref{q_z}, updating $q^{*}(\bZ)$ requires the expectations of other unknown parameters $\bpi, \bmu, \bLam$. In practice, we therefore need to estimate the variational approximation of these posterior distributions iteratively. Although the smoothing parameter $\omega$ does not affect the update of $q^{*}(\bZ)$ due to the power of $q(\bZ)$ in \eqref{elbo_gmm}, it remains crucial when evaluating the corresponding Evidence Lower Bound (ELBO).

    \item \textbf{Update of $q\left(\bpi, \bmu, \bLam\right)$: }

    Similarly, the log-posterior for the remaining parameters is given by
    \begin{equation}\label{q_other}
    \begin{aligned}
        \log q^{*}\left(\bpi, \bmu, \bLam\right) &= \mathbb{E}_{\bZ}\left\{\log p\left(\bX, \bZ, \bpi, \bmu, \bLam\right)\right\} + C \\
        &= \mathbb{E}_{\bZ}\left\{\log p^{\omega}\left(\bX | \bZ, \bmu, \bLam\right) + \log p^{\omega}\left(\bZ | \bpi\right) + \log p\left(\bpi\right) + \log p\left(\bmu | \bLam\right) + \log p\left(\bLam\right)\right\} + C \\
        &= \underbrace{\mathbb{E}_{\bZ}\left\{\log p^{\omega}\left(\bZ | \bpi\right)\right\} + \log p(\bpi)}_{\text{for $\bpi$}} + \underbrace{\mathbb{E}_{\bZ}\left\{\log p^{\omega}\left(\bX | \bZ, \bmu, \bLam\right)\right\} + \sum_{k = 1}^{K}\log p\left(\bmu_{k}, \bLam_{k}\right)}_{\text{for $(\bmu, \bLam)$}} + C.
    \end{aligned}
    \end{equation}
    From equation \eqref{q_other}, the variational approximation of $\bpi, \bmu, \bLam$ can be factorized as the variational approximation of $\bpi$ and the variational approximation of $(\bmu, \bLam)$. First, for $q^{*}(\bpi)$, we have 
    \begin{equation}
    \begin{aligned}
        \log q^{*}\left(\bpi\right) &= \mathbb{E}_{\bZ}\left\{\log p^{\omega}\left(\bZ | \bpi\right)\right\} + \log p(\bpi) + C \\
        &= \log \prod_{k = 1}^{K} \pi_{k}^{\alpha_{0} - 1} + \mathbb{E}_{\bZ}\left\{\omega\sum_{n = 1}^{N}\sum_{k = 1}^{K} z_{nk} \log \pi_{k}\right\} + C \\
        &= (\alpha_{0} - 1)\sum_{k = 1}^{K} \log \pi_{k} + \omega\sum_{n = 1}^{N}\sum_{k = 1}^{K}r_{nk} \log \pi_{k} + C, 
    \end{aligned}
    \end{equation}
    which yields 
    \begin{equation}\label{q_pi}
        \quad q^{*}(\bpi) = \prod_{k = 1}^{K} \pi_{k}^{\alpha_{0} - 1 + \omega\sum\limits_{n = 1}^{N}r_{nk}} + C.
    \end{equation}
    Therefore, $$q^{*}(\bpi) \sim \mathrm{Dir}(\bpi | \bal),$$
    where $\alpha_{k} = \alpha_{0} + \omega \cdot N_{k}$ with $N_k = \sum\limits_{n = 1}^{N}r_{nk}$.

    Next, to update the variational approxmiation $q^{*}(\bmu, \bLam)$ (or $q^{*}(\bmu_{k}, \bLam_{k})$), we have
    \begin{equation}
    \begin{aligned}
        \log q^{*}\left(\bmu, \bLam\right) &= \mathbb{E}_{\bZ}\left\{\log p^{\omega}\left(\bX | \bZ, \bmu, \bLam\right)\right\} + \sum_{k = 1}^{K}\log p\left(\bmu_{k}, \bLam_{k}\right) + C \\
        &= \sum_{k = 1}^{K}\sum_{n = 1}^{N} \mathbb{E}_{\bZ}\left\{z_{nk}\right\}\left(\dfrac{\omega}{2}\log |\bLam_{k}| - \dfrac{\omega}{2}(\bx_{n} - \bmu_{k})^{\prime} \bLam_{k} (\bx_{n} - \bmu_{k})\right) \\
        &\quad + \sum_{k = 1}^{K}\log p\left(\bmu_{k} | \bLam_{k}\right) + \sum_{k = 1}^{K} p\left(\bLam_{k}\right) + C \\
        &= \sum_{k = 1}^{K}\left\{\sum_{n = 1}^{N} \mathbb{E}_{\bZ}\left\{z_{nk}\right\} \left(\dfrac{\omega}{2}\log |\bLam_{k}| - \dfrac{\omega}{2}(\bx_{n} - \bmu_{k})^{\prime} \bLam_{k} (\bx_{n} - \bmu_{k})\right)\right.\\
        & \left.\quad\quad\quad + \left(\dfrac{1}{2} \log|\bLam_{k}| - \dfrac{1}{2}(\bmu_{k} - \bm_{0})^{\prime} \beta_{0}\bLam_{k} (\bmu_{k} - \bm_{0})\right) \right.\\
        & \left.\quad\quad\quad + \left(\dfrac{1}{2}(\nu_{0} - p - 1) \log |\bLam_{k}| - \dfrac{1}{2}\mathrm{tr}(\bW_{0}^{-1}\bLam_{k})\right) + C\right\},
    \end{aligned}
    \end{equation}
    and
    \begin{equation}\label{q_ml}
    \begin{aligned}
        \log q^{*}\left(\bmu_{k}, \bLam_{k}\right) &= \sum_{n = 1}^{N} \mathbb{E}_{\bZ}\left\{z_{nk}\right\} \left(\dfrac{\omega}{2}\log |\bLam_{k}| - \dfrac{\omega}{2}(\bx_{n} - \bmu_{k})^{\prime} \bLam_{k} (\bx_{n} - \bmu_{k})\right)\\
        & \quad + \left(\dfrac{1}{2} \log|\bLam_{k}| - \dfrac{1}{2}(\bmu_{k} - \bm_{0})^{\prime} \beta_{0}\bLam_{k} (\bmu_{k} - \bm_{0})\right) \\
        & \quad + \left(\dfrac{1}{2}(\nu_{0} - p - 1) \log |\bLam_{k}| - \dfrac{1}{2}\mathrm{tr}(\bW_{0}^{-1}\bLam_{k})\right) + C.
    \end{aligned}
    \end{equation}
    For $q^{*}\left(\bmu_{k}, \bLam_{k}\right)$, by equation \eqref{q_ml} and letting $\bar{\bx}_{k} = \dfrac{1}{N_{k}}\sum\limits_{n = 1}^{N}r_{nk}\bx_{n}, S_{k} = \dfrac{1}{N_{k}}\sum\limits_{n = 1}^{N} r_{nk}(\bx_{n} - \bar{\bx}_{k})(\bx_{n} - \bar{\bx}_{k})^{\prime}$, we have 
    \begin{equation}\label{q_mu}
        \log q^{*}\left(\bmu_{k} | \bLam_{k}\right) = -\dfrac{1}{2} \bmu_{k}^{\prime}\left(\beta_{0} + \omega N_{k}\right) \bLam_{k}\bmu_{k} + \bmu_{k}^{\prime} \bLam_{k}\left(\beta_{0}\bm_{0} + \omega N_{k}\bar{\bx}_{k}\right) + C.
    \end{equation}
    So $q^{*}\left(\bmu_{k} | \bLam_{k}\right) \sim N\left(\bmu_{k} | \bm_{k}, (\beta_{k}\bLam_{k})^{-1}\right)$, where the updated hyperparameters can be obtained via equation \eqref{q_mu} as $\beta_{k} = \beta_{0} + \omega \cdot N_{k}$ and $\bm_{k} = \dfrac{1}{\beta_{k}}\left(\beta_{0}\bm_{0} + \omega \cdot N_{k}\bar{\bx}_{k}\right)$. Finally, for $q^{*}\left(\bLam_{k}\right)$, by equations \eqref{q_ml} and \eqref{q_mu} as well as the fact that $\log q^{*}\left(\bLam_{k}\right) = \log q^{*}\left(\bmu_{k}, \bLam_{k}\right) - \log q^{*}\left(\bmu_{k} | \bLam_{k}\right)$, we have
    \begin{equation}\label{q_lam}
    \begin{aligned}
        \log q^{*}\left(\bLam_{k}\right) &= \dfrac{1}{2}(\bmu_{k} - \bm_{0})^{\prime}\beta_{0}\bLam_{k}(\bmu_{k} - \bm_{0}) + \dfrac{1}{2}\log |\bLam_{k}| - \dfrac{1}{2}\mathrm{tr}\left(\bW_{0}^{-1}\bLam_{k}\right)\\
        & \quad + \dfrac{1}{2}(\nu_{0} - p - 1) \log |\bLam_{k}| - \dfrac{1}{2}\sum_{n = 1}^{N} \omega r_{nk} (\bx_{n} - \bmu_{k})^{\prime}\bLam_{k}(\bx_{n} - \bmu_{k}) \\
        & \quad + \dfrac{1}{2}\sum_{n = 1}^{N} \omega r_{nk} \log |\bLam_{k}| + \dfrac{\beta_{k}}{2}(\bmu_{k} - \bm_{k})^{\prime}\bLam_{k}(\bmu_{k} - \bm_{k}) - \dfrac{1}{2} \log |\bLam_{k}| + C \\
        &= \dfrac{1}{2}\mathrm{tr}\bigg\{\bigg(\bW_{0}^{-1} + \beta_{0}(\bmu_{k} - \bm_{0})(\bmu_{k} - \bm_{0})^{\prime} + \sum_{n = 1}^{N} \omega r_{nk} (\bx_{n} - \bmu_{k})(\bx_{n} - \bmu_{k})^{\prime}\\
        &\quad\quad\quad - \beta_{k}(\bmu_{k} - \bm_{k})(\bmu_{k} - \bm_{k})^{\prime}\bigg) \bLam_{k} \bigg\} + \dfrac{1}{2}\left(\nu_{0} + \sum_{n = 1}^{N}\omega r_{nk} - p - 1\right)\log |\bLam_{k}| + C.
    \end{aligned}
    \end{equation}
    Therefore, the variational approximation $q^{*}\left(\bLam_{k}\right) \sim \mathcal{W}\left(\bLam_{k} | \bW_{k}, \nu_{k}\right)$, where
    \begin{equation}
    \begin{aligned}
        & \nu_{k} = \nu_{0} + \omega \cdot \sum_{n = 1}^{N} r_{nk}\\
        & \bW_{k} = \left(\bW_{0}^{-1} + \omega \cdot N_{k} \bS_{k} + \dfrac{\beta_{0}N_{k}}{\beta_{0} + N_{k}}(\bar{\bx}_{k} - \bm_{0})(\bar{\bx}_{k} - \bm_{0})^{\prime}\right)^{-1}.
    \end{aligned}
    \end{equation}

    \item \textbf{ELBO of GMM:}

    First, for \eqref{q_z}, according to \cite{bishop2006pattern}, we have
    \begin{equation}\label{elbo:gmm_exp}
    \begin{aligned}
        & \mathbb{E}_{\bpi}\left[\log \pi_{k}\right] = \psi(\alpha_{k}) - \psi\left(\sum_{k = 1}^{K} \alpha_{k}\right)\\
        & \mathbb{E}_{\bLam}\left[\log |\bLam_{k}|\right] = \sum_{j = 1}^{p} \psi\left(\dfrac{\nu_{k} + 1 - j}{2}\right) + p \log 2 + \log |\bW_{k}|\\
        & \mathbb{E}_{\bmu, \bLam}\left[(\bx_{n} - \bmu_{k})^{\prime}\bLam_{k}(\bx_{n} - \bmu_{k})\right] = p \beta_{k}^{-1} + \nu_{k}(\bx_{n} - \bm_{k})^{\prime}\bW_{k}(\bx_{n} - \bm_{k}),
    \end{aligned}
    \end{equation}
    where $\psi(\cdot)$ in \eqref{elbo:gmm_exp} is digamma function. Then, the ELBO of this model defined in \eqref{elbo_gmm} can be evaluated by
    \begin{equation}\label{elbo:gmm_eva}
    \begin{aligned}
        \mathrm{ELBO} &= \int_{\Theta} q(\bpi, \bmu, \bLam)q(\bZ) \log \left(\dfrac{p\left(\bX, \bpi, \bZ, \bmu, \bLam\right)}{q(\bpi, \bmu, \bLam)q^{\omega}(\bZ)}\right) d\bth\\
        &= \mathbb{E}_{q}\{\log p^{\omega}\left(\bX | \bZ, \bmu, \bLam\right) + \log p^{\omega}\left(\bZ | \bpi\right) + \log p\left(\bpi\right) + \log p\left(\bmu | \bLam\right) + \log p\left(\bLam\right) -\\
        & \quad \log q\left(\bpi\right) - \log q\left(\bmu, \bLam\right) - \log q^{\omega}\left(\bZ\right)\}.  
    \end{aligned}
    \end{equation}
    For each term in \eqref{elbo:gmm_eva}, we have
    \begin{equation}\label{elbo:gmm_all}
    \begin{aligned}
         \mathbb{E}\left\{\log p^{\omega}\left(\bX | \bZ, \bmu, \bLam\right)\right\} &= \dfrac{\omega}{2}\sum_{k = 1}^{K} N_{k}\bigg(\mathbb{E}\left[\log |\bLam_{k}|\right] - p \beta_{k}^{-1} - \nu_{k}\mathrm{Tr}(\bS_{k}\bW_{k}) - \\
        & \quad \nu_{k}(\bar{\bx}_{k} - \bm_{k})^{\prime}\bW_{k}(\bar{\bx}_{k} - \bm_{k}) - p\log(2\pi)\bigg)\\
        \mathbb{E}\left[\log p^{\omega}\left(\bZ | \bpi\right)\right] &= \omega \sum_{n = 1}^{N}\sum_{k = 1}^{K}r_{nk}\mathbb{E}\left[\log \pi_{k}\right]\\
        \mathbb{E}\left[\log p\left(\bpi\right)\right] &= \log \Gamma(K \alpha_{0}) - K\log \Gamma(\alpha_{0}) + (\alpha_{0} - 1)\sum_{k = 1}^{K} \mathbb{E}\left[\log \pi_{k}\right]\\
        \mathbb{E}\left[\log q^{\omega}\left(\bZ\right)\right] &= \omega \sum_{n = 1}^{N}\sum_{k = 1}^{K} r_{nk}\log r_{nk}\\
        \mathbb{E}\left[\log q\left(\bpi\right)\right] &= \log \Gamma\left(\sum_{k = 1}^{K} \alpha_{k}\right) - \sum_{k = 1}^{K}\log \Gamma(\alpha_{k}) + \sum_{k = 1}^{K} (\alpha_{k} - 1)\mathbb{E}\left[\log \pi_{k}\right]\\
        \mathbb{E}\left[\log p\left(\bmu, \bLam\right)\right] &= \dfrac{1}{2}\sum_{k = 1}^{K}\left\{p \log \dfrac{\beta_{0}}{2\pi} + \mathbb{E}\left[\log |\bLam_{k}|\right] - \dfrac{p\beta_{0}}{\beta_{k}} - \beta_{0}\nu_{k}(\bm_{k} - \bm_{0})^{\prime} \bW_{k} (\bm_{k} - \bm_{0})\right\}\\
        & \quad + \dfrac{\nu_{0} - p - 1}{2}\sum_{k = 1}^{K}\mathbb{E}\left[\log |\bLam_{k}|\right] - \dfrac{1}{2}\sum_{k = 1}^{K} \nu_{k} \mathrm{Tr}\left(\bW_{0}^{-1}\bW_{k}\right)\\
        & \quad + K\log |\bW_{0}|^{-\frac{\nu_{0}}{2}} \left(2^{\frac{\nu_{0}p}{2}} \pi^{\frac{p(p-1)}{4}} \prod_{j = 1}^{p} \Gamma\left(\dfrac{\nu_{0} - j + 1}{2}\right)\right)^{-1}\\
        \mathbb{E}\left[\log q\left(\bmu, \bLam\right)\right] &= \sum_{k = 1}^{K}\bigg\{\dfrac{1}{2} \mathbb{E}\left[\log |\bLam_{k}|\right] + \dfrac{p}{2} \log \dfrac{\beta_{k}}{2\pi} - \dfrac{p}{2} + \dfrac{\nu_{k} - p - 1}{2}  \mathbb{E}\left[\log |\bLam_{k}|\right] - \dfrac{\nu_{k}p}{2} \\
        & \quad + \log |\bW_{k}|^{-\frac{\nu_{k}}{2}} \left(2^{\frac{\nu_{k}p}{2}} \pi^{\frac{p(p-1)}{4}} \prod_{j = 1}^{p} \Gamma\left(\dfrac{\nu_{k} - j + 1}{2}\right)\right)^{-1}\bigg\}.
    \end{aligned}
    \end{equation}
    This completes the derivation of Algorithm \ref{CAVI_gmm}.
\end{itemize}

\subsection{Derivation of Algorithm \ref{CAVI_bmlr}}

For Bayesian mixture linear regression (BMLR) modeled in \eqref{CAVI_bmlr}, we update each term of $q$ in the mean-field family based on updating rule \eqref{CAVI} as follows. 

\begin{itemize}[leftmargin=*]
    \item \textbf{Update of $q(\bZ)$:} According to \eqref{CAVI}, \eqref{joint_bmlr} and \eqref{elbo_bmlr}, we have 
    \begin{equation}\label{bmlr:q_z}
    \begin{aligned}
        \omega \log q^{*}(\bZ) &= \mathbb{E}_{\bpi, \bet, \btau}\left[\log p\left(\bY, \bZ, \bpi, \bet, \btau | \mathcal{X}\right)\right] + C\\
        &= \mathbb{E}_{\bet}\left[\omega \log p\left(\bY | \bZ, \bet, \mathcal{X}\right)\right] + \mathbb{E}_{\bpi}\left[\omega \log p\left(\bZ | \bpi\right)\right] + C\\
        &= \omega\sum_{n = 1}^{N}\sum_{k = 1}^{K}z_{nk}\mathbb{E}_{\bet_{k}}\left[\log N\left(\by_{n} | \bX_{n}\bet_{k}, \lambda^{-1}\bI_{n}\right)\right] + \omega\sum_{n = 1}^{N}\sum_{k = 1}^{K}z_{nk}\mathbb{E}_{\pi_{k}}\left[\log \pi_{k}\right] + C\\
        &= \omega\sum_{n = 1}^{N}\sum_{k = 1}^{K}z_{nk} \underbrace{\left\{\mathbb{E}\left[-\dfrac{\lambda}{2} (\by_{n} - \bX_{n}\bet_{k})^{\prime}(\by_{n} - \bX_{n}\bet_{k})\right] + \mathbb{E}\left[\log \pi_{k}\right]\right\}}_{\log\rho_{nk}} + C. 
    \end{aligned}
    \end{equation}
    Therefore, we have $q^{*}(\bZ) \propto \prod\limits_{n = 1}^{N}\prod\limits_{k = 1}^{K} \rho_{nk}^{z_{nk}}$, and by normalizing $\rho_{nk}$ as $r_{nk} = \dfrac{\rho_{nk}}{\sum\limits_{j = 1}^{K} \rho_{nj}}$, we have $q^{*}(\bZ) = \prod\limits_{n = 1}^{N}\prod\limits_{k = 1}^{K} r_{nk}^{z_{nk}}$.

    \item \textbf{Update of $\log q^{*}(\btau)$:} Similarly, by \eqref{CAVI}, we have
    \begin{equation}\label{bmlr:q_tau}
    \begin{aligned}
        \log q^{*}(\btau) &= \mathbb{E}_{\bZ, \bpi, \bet}\left[\log p\left(\bY, \bZ, \bpi, \bet, \btau | \mathcal{X}\right)\right] + C\\
        &= \log p(\btau) + \mathbb{E}_{\bet}\left[\log p\left(\bet | \btau\right)\right] + C\\
        &= \sum_{k = 1}^{K}\left\{\mathbb{E}_{\bet_{k}}\left[\log p\left(\bet_{k} | \tau_{k}\right)\right] + \log p\left(\tau_{k}\right)\right\} + C.
    \end{aligned}
    \end{equation}
    Then, we have
    \begin{equation}
        \log q^{*}(\tau_{k}) = \left(a_{0} + \dfrac{p}{2} - 1\right)\log \tau_{k} - \left(b_{0} + \dfrac{1}{2}\mathbb{E}\left[\bet_{k}^{\prime}\bet_{k}\right]\right)\tau_{k},
    \end{equation}
    i.e. $q^{*}(\tau_{k}) = \mathrm{Gamma}(\tau_{k} | a_{k}, b_{k})$, where $a_{k} = a_{0} + \dfrac{p}{2}$ and $b_{k} = b_{0} + \dfrac{1}{2}\mathbb{E}[\bet_{k}^{\prime}\bet_{k}]$.

    \item \textbf{Update of $q^{*}\left(\bpi, \bet\right)$:} For $(\bpi, \bet)$, we have
    \begin{equation}\label{bmlr:q_pi_beta}
    \begin{aligned}
        \log q^{*}\left(\bpi, \bet\right) &= \mathbb{E}_{\bZ, \btau}\left[\log p\left(\bY, \bZ, \bpi, \bet, \btau | \mathcal{X}\right)\right] + C\\
        &= \underbrace{\omega \mathbb{E}_{\bZ}\left[\log p\left(\bY | \bZ, \bet, \mathcal{X}\right)\right] + \mathbb{E}_{\btau}\left[\log p\left(\bet | \btau\right)\right]}_{\text{for } \bet} + \underbrace{\omega\mathbb{E}_{\bZ}\left[\log p\left(\bZ | \bpi\right)\right] + \log p(\bpi)}_{\text{for } \bpi} + C.
    \end{aligned}
    \end{equation}
    Collecting the terms that depend on $\bpi$, we obtain
    \begin{equation}\label{bmlr:q_pi}
    \begin{aligned}
        \log q^{*}(\bpi) &= \omega\mathbb{E}_{\bZ}\left[\log p\left(\bZ | \bpi\right)\right] + \log p(\bpi) + C\\
        &= \sum_{k = 1}^{K}\log \pi_{k}^{\alpha_{0} - 1} + \omega\sum_{k = 1}^{K}\sum_{n = 1}^{N}r_{nk}\log \pi_{k} + C\\
            &= \sum_{k=1}^{K}\Bigl(\alpha_{0} - 1 + \omega\sum_{n=1}^{N}r_{nk}\Bigr)\log \pi_{k} + C.
                \end{aligned}
    \end{equation}
     This implies that $q^{*}(\bpi)$ follows a Dirichlet distribution
     $$q^{*}(\bpi) = \mathrm{Dir}(\bpi | \bal), \quad \text{where } \alpha_{k} = \alpha_{0} + \omega\sum\limits_{n = 1}^{N}r_{nk}.$$

    For $q^{*}(\bet_{k})$, we similarly collect the terms that depend on $\bet_k$:
    \begin{equation}\label{bmlr:q_beta}
    \begin{aligned}
        \log q^{*}(\bet_{k}) &= \omega \mathbb{E}_{\bZ}\left[\log p\left(\bY | \bZ, \bet, \mathcal{X}\right)\right] + \mathbb{E}_{\btau}\left[\log p\left(\bet | \btau\right)\right] + C\\
        &= \omega\sum_{n = 1}^{N}r_{nk}\left\{-\dfrac{\lambda}{2} (\by_{n} - \bX_{n}\bet_{k})^{\prime}(\by_{n} - \bX_{n}\bet_{k})\right\} - \dfrac{1}{2}\mathbb{E}_{\tau_{k}}\left[\tau_{k}\right]\bet_{k}^{\prime}\bet_{k} + C\\
        &= \lambda\bet_{k}^{\prime}\omega\sum_{n = 1}^{N}r_{nk}\bX_{n}^{\prime}\by_{n} - \dfrac{1}{2}\bet_{k}^{\prime}\left\{\mathbb{E}_{\tau_{k}}[\tau_{k}]\bI_{p} + \lambda\omega\sum_{n = 1}^{N}r_{nk}\bX_{n}^{\prime}\bX_{n}\right\}\bet_{k} + C.
            \end{aligned}
    \end{equation}
    Thus, 
         $q^{*}(\bet_{k}) = N\left(\bet_{k} | \bm_{k}, \bS_{k}\right),$ where $$\left\{\begin{array}{l} \bS_{k}^{-1} = \left(\mathbb{E}_{\tau_{k}}[\tau_{k}]\bI_{p} + \omega \cdot \lambda\sum\limits_{n = 1}^{N}r_{nk}\bX_{n}^{\prime}\bX_{n}\right)^{-1}\\
        \bm_{k} = \omega \cdot \lambda\bS_{k}\sum\limits_{n = 1}^{N}r_{nk}\bX_{n}^{\prime}\by_{n}\end{array}\right..$$

    \item \textbf{ELBO of BMLR:} Finally, we provide the evidence lower bound of Bayesian mixture linear regression which is formulated as \eqref{elbo_bmlr}. According to \eqref{elbo_bmlr}, we have
    \begin{equation}\label{elbo:bmlr_eva}
    \begin{aligned}
        \mathrm{ELBO} &= \int_{\Theta} q(\bZ)q(\bpi, \bet)q(\btau) \log \left(\dfrac{p\left(\bY, \bZ, \bpi, \bet, \btau | \mathcal{X}\right)}{q(\bpi, \bet)q(\btau)q^{\omega}(\bZ)}\right) d\bth\\
        &= \mathbb{E}_{q}\left[\log p^{\omega}\left(\bY | \bZ, \bet, \mathcal{X}\right) + \log p^{\omega}\left(\bZ | \bpi\right) + \log p\left(\bpi\right) + \log p\left(\bet | \btau\right) + \log p\left(\btau\right)\right.\\
        &\left.\quad - \log q\left(\bpi, \bet\right) - \log q\left(\btau\right) - \log q^{\omega}\left(\bZ\right)\right].
    \end{aligned}
    \end{equation}
    Each term in equation \eqref{elbo:bmlr_eva} is
    \begin{equation}\label{elbo:bmlr_all}
    \begin{aligned}
        \mathbb{E}\left[\log p^{\omega}\left(\bY | \bZ, \bet, \mathcal{X}\right)\right] &= \omega\sum_{n = 1}^{N}\sum_{k = 1}^{K}r_{nk}\left\{-\dfrac{1}{2}\log 2\pi\lambda^{-1}\bI_{n} - \dfrac{\lambda}{2}\by_{n}^{\prime}\by_{n} + \lambda \bm_{k}^{\prime}\bX_{n}^{\prime}\by_{n}\right.\\
        &\left.\quad - \dfrac{\lambda}{2}\mathrm{Tr}\left(\bX_{n}^{\prime}\bX_{n}\left(\bm_{k}\bm_{k}^{\prime} + \bS_{k}\right)\right)\right\}\\
        \mathbb{E}\left[\log p^{\omega}\left(\bZ | \bpi\right)\right] &= \omega\sum_{n = 1}^{N}\sum_{k = 1}^{K}r_{nk}\left(\psi(\alpha_{k}) - \psi\left(\sum\alpha_{k}\right)\right)\\
        \mathbb{E}\left[\log p\left(\bpi\right)\right] &= \log \Gamma(K \alpha_{0}) - K\log \Gamma(\alpha_{0}) + (\alpha_{0} - 1)\sum_{k = 1}^{K}\left(\psi(\alpha_{k}) - \psi\left(\sum\alpha_{k}\right)\right)\\
        \mathbb{E}\left[\log p\left(\bet | \btau\right)\right] &= \sum_{k = 1}^{K}\left(-\dfrac{D}{2}\log 2\pi + \dfrac{D}{2}\left(\psi(a_{k}) - \log b_{k}\right) - \dfrac{a_{k}}{2b_{k}}\left(\bm_{k}^{\prime}\bm_{k} + \mathrm{Tr}(\bS_{k})\right)\right)\\
        \mathbb{E}\left[\log p\left(\btau\right)\right] &= \sum_{k = 1}^{K}\left(a_{0}\log b_{0} + (a_{0} - 1)\left(\psi(a_{k}) - \log b_{k}\right) - \dfrac{b_{0}a_{k}}{b_{k}} - \log \Gamma(a_{0})\right)\\
        \mathbb{E}\left[\log q^{\omega}\left(\bZ\right)\right] &= \omega\sum_{n = 1}^{N}\sum_{k = 1}^{K} r_{nk}\log r_{nk}\\
        \mathbb{E}\left[\log q\left(\bpi\right)\right] &= \log\Gamma\left(\sum\alpha_{k}\right) - \sum_{k = 1}^{K}\log\Gamma(\alpha_{k}) + \sum_{k = 1}^{K}(\alpha_{k} - 1)\left(\psi(\alpha_{k} - \psi\left(\sum\alpha_{k}\right)\right)\\
        \mathbb{E}\left[\log q\left(\bet\right)\right] &= \sum_{k = 1}^{K}\left(-\dfrac{1}{2}\log|\bS_{k}| - \dfrac{p}{2}\log(2\pi e)\right)\\
        \mathbb{E}\left[\log q\left(\bpi\right)\right] &= \sum_{k = 1}^{K}\left(-\log\Gamma(a_{k}) + \log b_{k} + (a_{k} - 1)\psi(a_{k}) - a_{k}\right).
    \end{aligned}
    \end{equation}
    This completes the derivation of Algorithm \ref{CAVI_bmlr}.
\end{itemize}

\section{Comparative Study}
\label{sec:supp2}
In Appendix \ref{sec:supp2}, we present a comparative study between General Posterior Calibration (GPC) and TVB with sequential updates for both the Gaussian mixture model (GMM) and the Bayesian mixture linear regression (BMLR) examples discussed in Section \ref{empirical}. For the GMM example, we maintain the same settings as described in Section \ref{numerical_gmm}. Similarly, for the BMLR example, we use the same setup as Section \ref{numerical_bmlr}, but with $N = 1000$ and $J_{n} = (2000, 3000, 5000)$. The results are presented below.

\begin{figure}[ht!]
    \centering
    \includegraphics[width = 14cm]{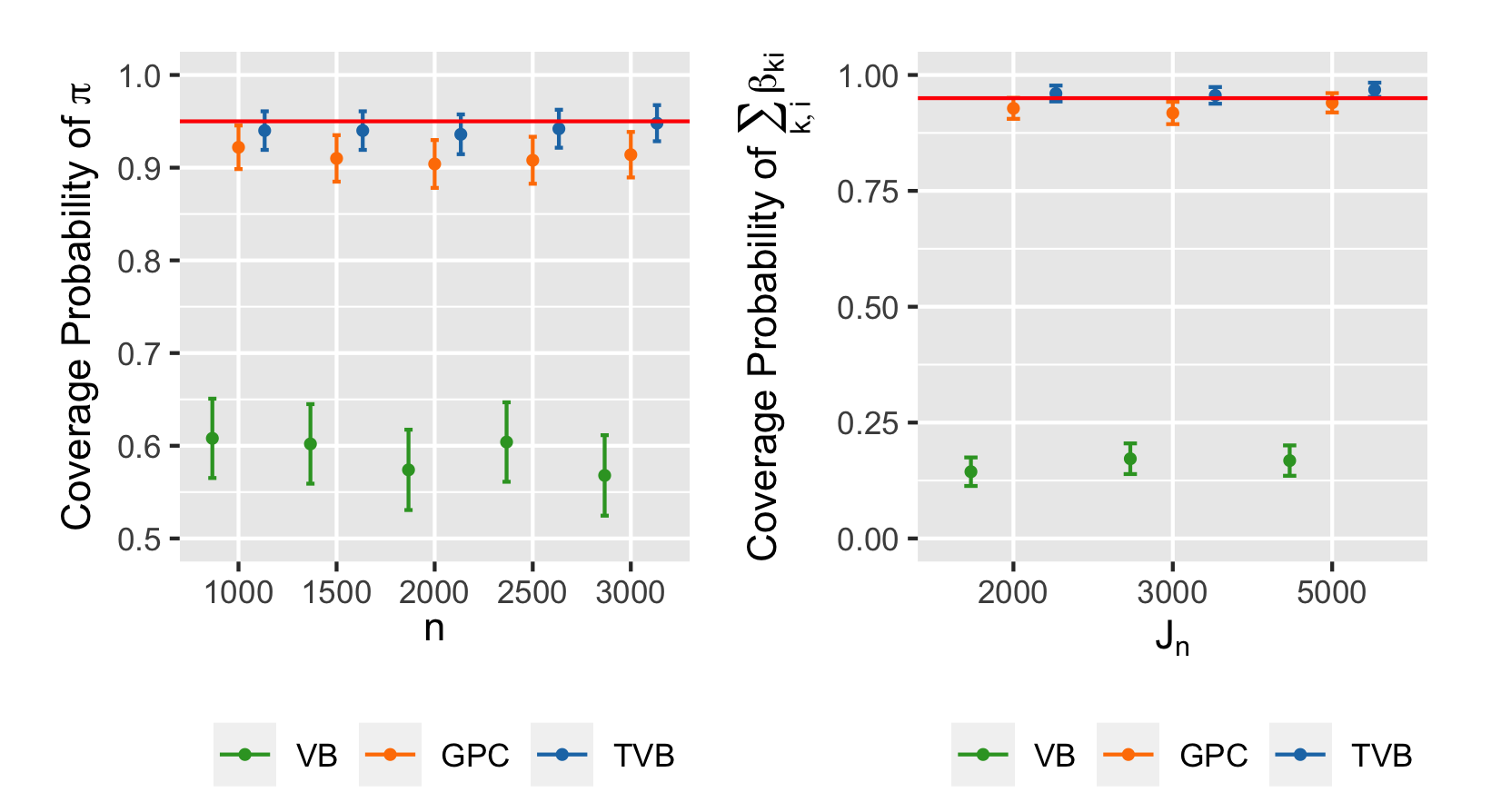}
    \caption{Frequentist coverage performance for original VB method (VB), general posterior calibration method (GPC) and trustworthy VB method (TVB) with sequential update. The left panel displays the coverage of the mixing parameter $\pi$ in the GMM example. The right panel presents the coverage of the sum of regression coefficients, $\sum\limits_{k=1}^{K}\sum\limits_{i=1}^{p}\beta_{ki}$ in the BMLR example. In each panel, dot plots with error bars indicate the average coverage and the 95\% confidence interval.}
    \label{fig:plot_cover_app}
\end{figure}

\begin{figure}[ht!]
    \centering
    \includegraphics[width = 14cm]{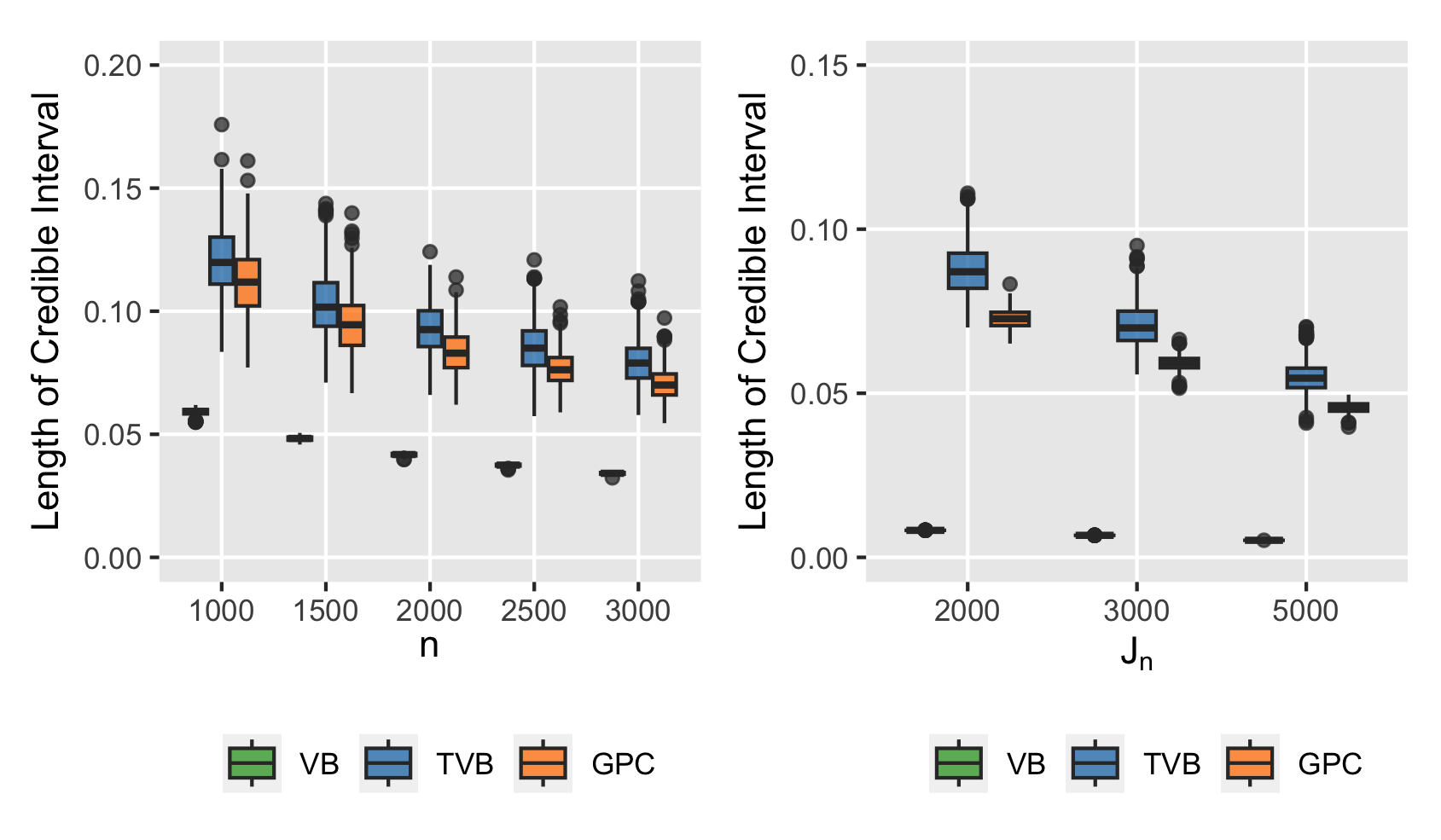}
    \caption{Boxplots of credible interval length for original VB method (VB), general posterior calibration method (GPC) and trustworthy VB method (TVB) with sequential update. The left panel is for $\pi$ in the GMM example; the right panel is for $\sum\limits_{k = 1}^{K}\sum\limits_{i = 1}^{p}\beta_{ki}$ in the BMLR example.}
    \label{fig:plot_length_app}
\end{figure}

\begin{figure}[ht!]
    \centering
    \includegraphics[width = 14cm]{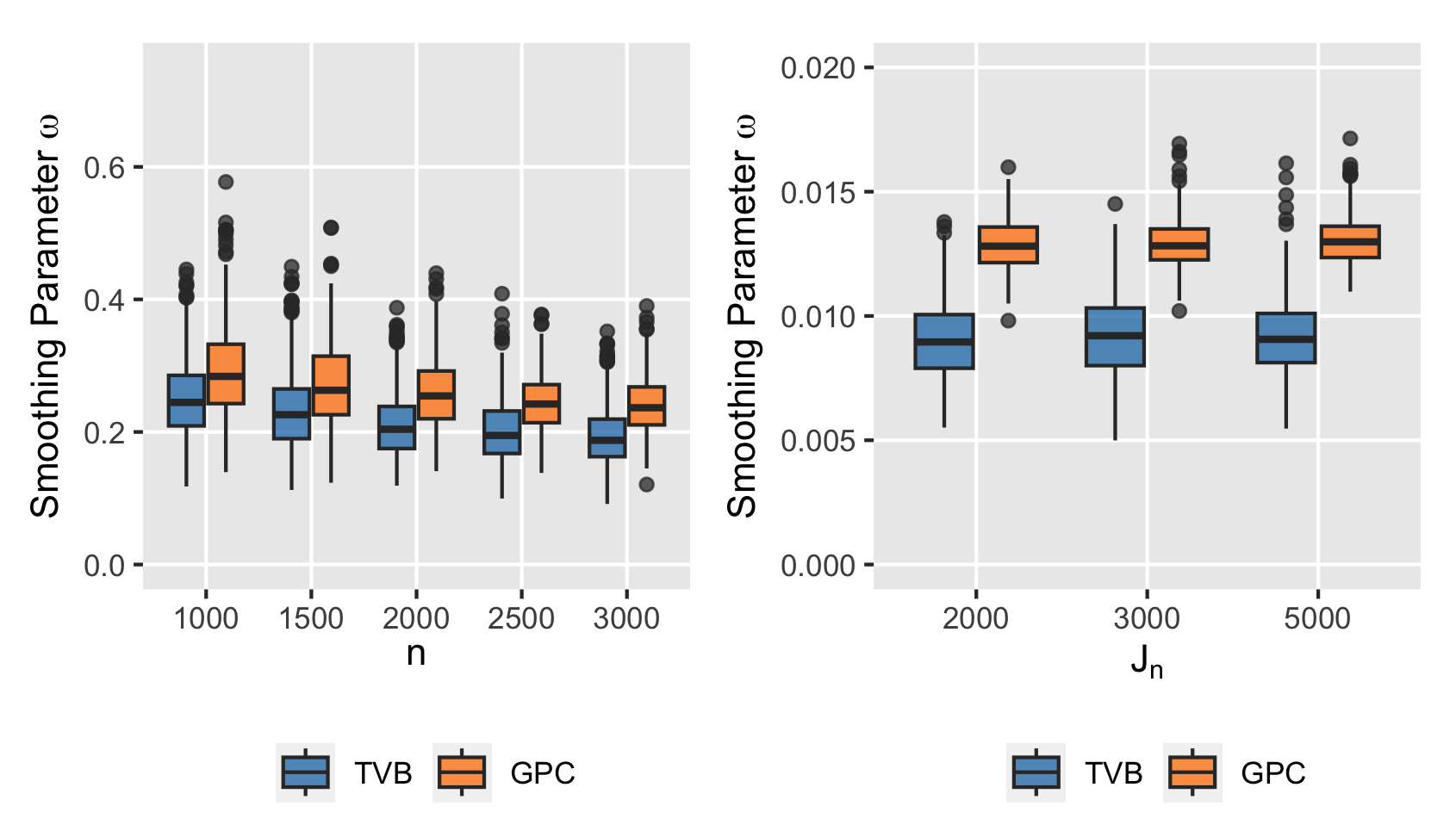}
    \caption{Boxplots of the selected smoothing parameter $\omega$ for general posterior calibration method (GPC) and trustworthy VB method (TVB) with sequential update. The left panel is for the GMM example, while the right panel is for the BMLR example.}
    \label{fig:plot_omega_app}
\end{figure}

As shown in Figure \ref{fig:plot_cover_app}, the TVB method achieves better accuracy and robustness compared to the GPC method, although GPC still improves calibration over standard VB. As expected, the TVB method produces wider credible intervals than GPC, while utilizing a smaller smoothing parameter $\omega$; see Figures \ref{fig:plot_length_app} and \ref{fig:plot_omega_app}, respectively.

\end{document}